\documentclass[11pt,a4paper,UKenglish,cleveref, autoref]{article}
\usepackage[OT1]{fontenc}
\usepackage[utf8]{inputenc}
\usepackage{graphicx}
\usepackage[dvipsnames]{xcolor}
\usepackage[margin=1in]{geometry}

\usepackage{amsthm,amssymb,amsmath}
\usepackage[sf,bf,small,raggedright,compact]{titlesec}
\usepackage{hyperref}
\usepackage{tcolorbox}
\definecolor{linkblue}{named}{MidnightBlue}
\hypersetup{colorlinks=true, linkcolor=linkblue,  anchorcolor=linkblue,
        citecolor=linkblue, filecolor=linkblue, menucolor=linkblue,
        urlcolor=linkblue}
\usepackage[capitalize]{cleveref}
\usepackage{paralist}
\usepackage[longnamesfirst,numbers,sort&compress]{natbib}

\usepackage{stmaryrd}

\usepackage{listings, newtxtt}
\newtheorem{thm}{Theorem}
\newtheorem{obs}[thm]{Observation}
\newtheorem{lem}[thm]{Lemma}

\crefname{lem}{Lemma}{Lemmata}
\newtheorem{cor}[thm]{Corollary}
\newtheorem{prop}[thm]{Proposition}
\newtheorem{clm}[thm]{Claim}
\newtheorem*{gapeth}{Gap Exponential Time Hypothesis (Gap-ETH)}
\newtheorem{question}{Question}

\newcommand{\R}{\mathbb{R}}
\newcommand{\var}{\mathsf{Var}}
\newcommand{\cla}{\mathsf{Cla}}
\newcommand{\emb}{\mathsf{Emb}}

\newcommand{\constraints}{\cC}
\newcommand{\poly}{\mathrm{poly}}

\newcommand{\eps}{\varepsilon}
\newcommand{\gcsp}{\textsc{Geometric} CSP}
\newcommand{\pcsp}{\textsc{Projection} CSP}
\newcommand{\MIS}{\textsc{Maximum Independent Set}}
\newcommand{\MIF}{\textsc{Maximum Induced Forest}}
\newcommand{\MIM}{\textsc{Maximum Induced Matching}}
\newcommand{\MPS}{\textsc{Minimum Piercing Set}}

\newcommand{\MDS}{\textsc{Minimum Dominating Set}}

\newcommand{\tcsp}{\textsc{Tiling} CSP}

\newcommand{\mati}{\textsc{Matrix Tiling}}
\newcommand{\bsumset}{\textsc{Binary-SumSet}}
\newcommand{\OPT}{\textsc{opt}}
\newcommand{\lecsp}{$\le$-CSP}

\newcommand{\IR}{\mathbb{R}}

\newcommand{\ve}[1]{{#1}}
\newcommand{\dir}[1]{\vec{{#1}}}
\newcommand{\boldsymbolvec}[1]{\vec{{#1}}}
\newcommand{\wires}{\cW}
\newcommand{\wirebound}{\kappa}
\newcommand{\wire}{W}
\newcommand{\domain}{D}
\newcommand{\loc}{\mathsf{loc}}
\newcommand{\totalconstraints}{\mathfrak{m}}
\newcommand{\gridconst}{c}
\newcommand{\phiv}{\mathsf{va}}
\newcommand{\phic}{\mathsf{cl}}

\newcommand{\etal}{\emph{et~al.}}

\newcounter{ctr}
\loop
 \stepcounter{ctr}
 \expandafter\edef\csname c\Alph{ctr}\endcsname{\noexpand\mathcal{\Alph{ctr}}}
\ifnum\thectr<26
\repeat

\date{}

\title{Shifting is Optimal under Gap-ETH: A Lower Bound Framework for Geometric Approximation Schemes}

\author{Manuel C\'{a}ceres\thanks{Department of Computer Science, Aalto University, Finland and Department of Mathematics and Computer Science, University of Southern Denmark, Denmark. Research supported by the Helsinki Institute for Information Technology HIIT.} \qquad S\'{a}ndor Kisfaludi-Bak\thanks{Department of Computer Science, Aalto University, Espoo, Finland. Research supported by the Research Council of Finland, Grant 363444.} \qquad Saeed Odak\footnotemark[2]}

\begin{document}

\maketitle

\begin{abstract}
    The shifting technique of Hochbaum and Maass [J.ACM'85] produces PTASes with the fastest known running times $n^{O(1/\varepsilon^{d-1})}$ for several $d$-dimensional geometric problems. However, it is only known, due to Marx [FOCS'07], that these algorithms are indeed optimal for dimension $d=2$.
    We show that these running times are optimal under Gap-ETH for every constant dimension. More precisely, we develop a framework that enables us to prove the conditional optimality of the shifting algorithms for several problems on unit ball graphs, such as maximum independent set, maximum induced forest, and others, as well as for the problem of piercing unit balls.
    Our framework is built using the cube wiring theorem of De Berg et al.\ [SICOMP'20] and the reduction steps of Marx and Sidiropoulos [SoCG'14] to create a convenient maximization version of geometric CSP that can be used as a basis for reductions.
\end{abstract}

\section{Introduction}

One of the great successes in the area of geometric algorithms is the existence of polynomial-time approximation schemes (PTASes) in low-dimensional settings. The resulting running times are usually of the form $2^{f(1/\eps)}\poly(n)$ or $n^{f(1/\eps)}$. The former are called \emph{efficient} polynomial time approximation schemes, or EPTAS for short. Once a PTAS or EPTAS running time is achieved for a problem, the research focuses on optimizing the dependence on $n$ and $1/\eps$ until a matching (conditional) lower bound is proven.

The first technique for proving EPTAS and PTAS lower bounds was introduced by Marx~\cite{Marx07}. Under the Exponential Time Hypothesis (ETH)~\cite{ImpagliazzoPZ01}, Marx provided almost matching lower bounds for several problems on planar graphs and on intersection graphs of unit disks: for maximum independent set, the running times $2^{O(1/\eps)}\poly{(n)}$ in planar graphs and $n^{O(1/\eps)}$ in unit disk graphs are essentially optimal. Using the stronger Gap-ETH hypothesis~\cite{Dinur16,ManurangsiR17}, Marx's techniques give \emph{tight} lower bounds: they match up to constant factors in the exponent. However, the technique does not extend to higher-dimensional geometric problems.

In dimension $3$ and above, a lower bound framework was devised by De Berg~\etal~\cite{deBergBKMZ20} for exact algorithms. These results can be extended by using their Cube Wiring Theorem~(see also~\cite{sandorPHDThesis}) and a careful analysis of gaps, to obtain lower bounds of the form $2^{\Omega(1/\eps^{d-1})}\poly(n)$ under Gap-ETH, in particular for the Euclidean traveling salesman problem~\cite{KisfaludiBakNW25}. The same technique has been used recently for proving EPTAS lower bounds for clustering problems~\cite{CohenAddadCSS2026}.

This higher-dimensional lower bound technique only produces lower bounds of the form $2^{\Omega(1/\eps^{d-1})}\poly(n)$. Currently, there is no technique that yields lower bounds of the form $n^{\Omega(1/\eps^{d-1})}$. On the other hand, the shifting technique of Hochbaum and Maass~\cite{HochbaumM85} is a common source of such running times. Shifting is a simple meta-algorithm that can be sketched as follows. Consider a randomly shifted grid consisting of hypercube cells of side length $O(1/\eps)$. Obtain a naive solution near the boundary of the cells. Due to the random shift, any given constant size ball has a probability of at most $O_d(\eps)$ of being intersected by a cell boundary, thus in expectation the naive solution introduces an error of at most $O_d(\eps \cdot \OPT)$. Additionally, the remaining instances in the interiors of the cells are pairwise disjoint and do not interfere, thus they can be solved separately. Each such instance has a volume of $O(1/\eps^d)$ and has geometric separators of size $O(1/\eps^{d-1})$, so one can solve all cell instances exactly in $n^{O(1/\eps^{d-1})}$ time using techniques from parameterized exact algorithms~\cite{CyganFKLMPPS15}. These algorithms can be derandomized with a polynomial overhead by simply checking all shifts. Problems that have such PTASes include maximum independent set~\cite{Chan03} and minimum dominating set~\cite{DeDCN13} in intersection graphs of unit balls. This raises a natural question.

\begin{question}\label{q:MIS}
Is the $n^{O(1/\eps^{d-1})}$-time shifting algorithm for \MIS{} in unit ball graphs optimal for every fixed $d\geq 3$?
\end{question}

We note that existing parameterized lower bounds rule out an EPTAS for \MIS. Under ETH, the problem has no $f(k)n^{o(k^{1-1/d})}$-time algorithm for any computable $f$~\cite{Marx06MDS,MarxS14} parameterized by the solution size $k$. As such, if an EPTAS exists, setting $\eps=1/(2k)$ derives an $f(1/\eps)\poly(n)=f(k)\poly(n)$-time algorithm, contradicting ETH. However, it is unclear how to use this result to obtain a tight lower bound.

There are two natural approaches to tackle \Cref{q:MIS}. One is to attempt to generalize Marx's technique~\cite{Marx07}. This technique~\cite{Marx07}, based on the so-called \mati{} problem. \mati{} is defined as follows. An instance of \mati{} consists of integers $k$, $D$, and $k^2$ nonempty sets $C_{i,j} \subseteq [D] \times [D]$ for $1 \leq i, j \leq k$. The task is to select in each cell $(i,j)$ of a $k \times k$ matrix some element $(a,b) \in C_{i,j}$, such that the first coordinates of the selected elements are equal in horizontally neighboring cells and the second coordinates are equal in vertical neighbors. The objective is to maximize the number of such equalities.\footnote{We note that the original definition of \mati{} in~\cite{Marx07} is slightly different, but equivalent to ours in terms of approximability.} See \Cref{fig:matrix.tilling}. Unfortunately, the natural extension of matrix tiling to $3$ dimensions, where we require neighboring cells to agree on the corresponding coordinate, seems to behave differently: in order to find a $(1-\eps)$-approximate solution of a satisfiable instance, we can follow a simple shifting strategy. We ignore cells that intersect a randomly shifted grid of side length $1/\eps$ (these will contribute to the error). The remaining cells inside each $1/\eps\times 1/\eps\times 1/\eps$ portion can be solved exactly in only $n^{O(1/\eps)}$ time: indeed, one can simply try all choices of values on the $O(1/\eps)$ axis-aligned slabs of cells that share one coordinate. As a result, the complexity of this problem in dimension $d$ is likely different from $n^{\Theta(1/\eps^{d-1})}$, and it is not suitable for our purposes.

\begin{figure}
        \centering
        \includegraphics[width=0.6\linewidth,page=1]{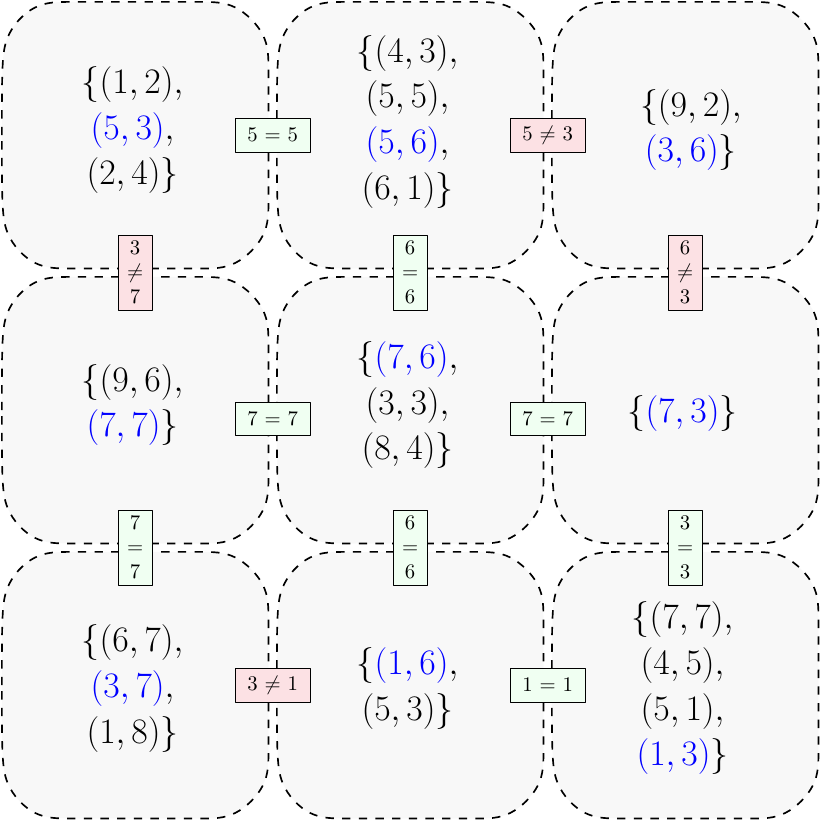}
        \caption{An instance of \mati{}: the cell assignments are shown in blue; the equality and conflicting assignment on the corresponding coordinates of adjacent cells are shown in green and red respectively.}
        \label{fig:matrix.tilling}
    \end{figure}

The second approach requires taking a step back. \textsc{Matrix Tiling} is the optimization variant of the \textsc{Grid Tiling} problem from the exact parameterized algorithms literature, which is convenient for proving lower bounds of the form $n^{\Omega(\sqrt{k})}$ for geometric problems in the plane~\cite{CyganFKLMPPS15}. A higher-dimensional generalization of this exact parameterized lower bound framework was proposed by Marx and Sidiropoulos~\cite{MarxS14}; the optimization variants of these problems could serve as the starting point for the PTAS lower bound framework. Marx and Sidiropoulos defined the following variants of constraint satisfaction (CSP) problems.

\subparagraph*{Constraint Satisfaction Problems.}
An instance $I$ of a \emph{binary constraint satisfaction problem} (CSP) is a triple $(V, D, C)$, where:
\begin{itemize}
    \item $V$ is a finite set of \textit{variables},
    \item $D$ is a \textit{domain} of values, and
    \item $\constraints$ is a set of \textit{constraints}, where it contains a unary constraint $C_\ve{v} \subseteq D$ for each $\ve{v} \in V$ and some binary constraints $C_{\ve{u},\ve{v}} \subseteq D \times D$ for distinct $\ve{u}, \ve{v} \in V$.
\end{itemize}
An \emph{assignment} of $I$ is a function $f : V \to \domain$. A unary constraint $C_{\ve{a}}$ is \emph{satisfied} by $f$ if $f(\ve{a}) \in C_{\ve{a}}$, and a binary constraint $C_{\ve{a},\ve{b}}$ is satisfied if $(f(\ve{a}), f(\ve{b})) \in C_{\ve{a},\ve{b}}$. An assignment is a \emph{solution} (or \emph{satisfying assignment}) if it satisfies every constraint in~$\constraints$. The \emph{primal graph} of a binary CSP instance $I = (V, D, C)$ is the undirected simple graph $G = (V, E)$ where distinct vertices $u, v \in V$ satisfy $\{u, v\} \in E$ if and only if there exists a binary constraint $C_{u,v} \in \constraints$. We call a pair of distinct variables $\ve{u}, \ve{v} \in V$ \emph{adjacent} if there is a binary constraint $C_{\ve{v},\ve{u}} \in \constraints$.

A \emph{$d$-dimensional geometric CSP} (\gcsp) instance $I = (V, D, C)$ is a binary CSP instance whose primal graph is isomorphic to an induced subgraph of the $d$-dimensional grid hypercube $G_k^d$ for some positive integer length $k$, via a subgraph isomorphism $\loc : V \to V(G_k^d)$. Note that the constraint set $\constraints$ includes a unary constraint $C_{\ve{v}} \subseteq D$ for each $\ve{v} \in V$ and a binary constraint $C_{\ve{u},\ve{v}} \subseteq D \times D$ for each edge $\{\ve{u},\ve{v}\}$ in the primal graph.

A \emph{$d$-dimensional tiling CSP} (\tcsp) is a \gcsp{} instance where the domain is $D = [\Lambda]^d$ for some positive integer $\Lambda$. Similarly to \gcsp, the constraint set $\constraints$ includes a unary constraint $C_{\ve{v}} \subseteq D$ for each $\ve{v} \in V$. However, the binary constraints are defined \emph{implicitly}: for every pair of adjacent variables $\ve{a}, \ve{b} \in V$ with $\loc(\ve{b}) = \loc(\ve{a}) + \dir{e}_i$ ($\dir{e}_i$ is the standard basis vector), the constraint $C_{\ve{a},\ve{b}} \subseteq D \times D$ contains exactly the pairs $((x_1,\dots,x_d),(y_1,\dots,y_d)) \in [\Lambda]^d \times [\Lambda]^d$ satisfying $x_i = y_i$. A \emph{$d$-dimensional $\le$-CSP} (\lecsp) is identical to \tcsp, except the binary constraints are of the form $x_i\leq y_i$ instead of $x_i=y_i$. A key difference from \mati{} is that the primal graph in \tcsp{} (or \gcsp{}) can be any induced subgraph of the grid.

In the optimization versions of all these CSPs (denoted with the prefix ``\textsc{Max-}''), the goal is to find an assignment that maximizes the number of satisfied constraints, or equivalently, minimizes the number of violated constraints.
Marx and Sidiropoulos~\cite{MarxS14} show that \tcsp{}, \lecsp{} and \textsc{Independent Set} in ball graphs has no $f(k)n^{O(k^{1-1/d})}$-time algorithm for any computable $f$ under ETH. It is therefore natural to ask if the maximization variants of these problems (maximizing the number of satisfied constraints) is a good basis for PTAS lower~bounds.

\subsection{Our contribution}

We provide a framework that derives tight conditional lower bounds for geometric problems solved with the shifting technique in any constant dimension $d\geq 3$. More precisely, we prove that the maximization variants of \tcsp{} and \lecsp{} have the desired lower bounds.

\begin{thm}[Main theorem]
    \label{thm:main-gcsp}
    Assuming Gap-ETH, for every integer $d \ge 2$ there exists a constant $\gamma > 0$ such that there is no algorithm with running time $N^{\gamma / \eps^{d-1}}$ that, given any $\eps>0$ and instance $I = (V,D,\constraints)$ of $d$-dimensional \tcsp{} or \lecsp{} of input size $N = \poly(|\constraints|,|D|)$, can distinguish between the following two cases:
    \begin{enumerate}
        \item[(i)] there exists an assignment satisfying all constraints, and
        \item[(ii)] every assignment violates at least an $\eps$-fraction of the constraints.
    \end{enumerate}
\end{thm}

We obtain the above theorem using a chain of reductions from the \textsc{Max 3-SAT} problem. The first step of the reduction chain in~\cite{MarxS14} involves a complicated embedding of the primal graph of the $3$-SAT formula into a grid, where maintaining the gap for partially satisfied formulas is difficult. Instead, we use a much simpler embedding via the Cube Wiring Theorem of De Berg~\etal~\cite{deBergBKMZ20}. We map variables and clauses of a formula $\phi$ to the $d$-dimensional grid graph and use a strong form of the Cube Wiring Theorem to find internally vertex-disjoint paths (wires) mimicking the incidence graph of $\phi$. Intuitively, these wires will ``carry'' the corresponding boolean value of an assignment of $\phi$. To ensure this behavior, we partition the grid graph into small cubes and assign to each cube a CSP variable whose value encodes the boolean values of the wires passing through that cube. We then use unary constraints to enforce consistency of literals and satisfaction of clauses, and binary constraints to enforce consistency of wire values across adjacent cubes. This reduction has the desired locality structure to ensure that a constant fraction of error in $\phi$ will propagate as a proportional fraction of error in the corresponding CSP instance.

For the rest of the steps in the reduction chain, we proceed analogously to~\cite{MarxS14}, but we need to create stronger gap-preserving reduction steps. This needs some work, but can be completed after observing important locality properties of the constructions at each step.

To prove the usefulness of our main theorem, we show that the shifting-based algorithms for a wide range of problems are Gap-ETH-tight. In particular, we consider several problems in intersection graphs of unit balls. Given a set of balls in $d$-dimensional space, we are looking for the maximum independent set (a maximum set of pairwise disjoint balls), a minimum dominating set (a minimum subset of balls where the union of their closed graph neighborhoods is all balls), a maximum induced matching (a maximum set of balls whose induced subgraph is a matching), or a maximum induced forest (a maximum set of balls whose induced subgraph has no cycle). We also consider the piercing set problem: given a set of balls, find the minimum set of points such that each ball contains a point. The following theorem demonstrates several applications of our framework, and gives an affirmative answer to Question~\ref{q:MIS} assuming Gap-ETH.

\begin{thm}[Applications] \label{thm:application}
    For any constant $d \ge 2$, \MIS{}, \MDS{}, \MIF{}, \MIM{}, \MPS{} have PTASes with running time $n^{O(1/\eps^{d-1})}$ on unit ball graphs in $\R^d$. Under Gap-ETH, for any constant $d \ge 2$ there exists some $\gamma>0$ such that none of these problems have an $n^{\gamma/\eps^{d-1}}$-time PTAS.
\end{thm}

The upper bounds of this theorem are either known or follow the standard shifting technique~\cite{HochbaumM85}, paired with divide-and-conquer algorithms in the hypercubes of side length $O(1/\eps)$. For the lower bounds, we use our framework: we reduce either from \textsc{Max-}\tcsp{} or \textsc{Max-}\lecsp{}. In case of \MIS{} and \MDS{} we are able to use existing gadgetry~\cite{Marx06MDS,MarxS14} without changes. For \MPS{} we rely on an intricate gadgetry introduced recently for the $k$-center problem by Blank~\etal~\cite{geert}. For \MIF{} and \MIM{} we introduce new gadgetry that may be of independent interest.

\subsection{Further related work}

The connection of approximation lower bounds and parameterized complexity has been studied before~\cite{GargK07,KulikS10}. More recently, Gap-ETH-based lower bounds have been an important tool in parameterized approximation algorithms, see~\cite{BonnetEKP15, BhattacharyyaGKM18,ChalermsookCKLM20} and the survey~\cite{FeldmannSLM20}.

Shifting as an algorithmic paradigm has several variants. In the setting of planar graphs, it was first introduced by Baker~\cite{Baker94}. A similar source of approximation schemes is the bidimensionality theory of Demaine and Hajiaghayi~\cite{DemaineH05}. In the Euclidean setting, sometimes it is important to perform shifting on a hierarchical decomposition rather than a fixed grid, leading to approximation schemes that tackle objects/distances at different scales simultaneously. Typically one shifts a so-called quadtree~\cite{Arora98, Chan03}. Our lower bound framework does provide lower bounds for some quadtree-based shifting algorithms, in particular we obtain a matching lower bound for Chan's algorithm~\cite{Chan03} for independent set in intersection graphs of balls (of arbitrary radius) in Euclidean space.

\begin{figure}
    \centering
    \includegraphics[width=0.95\linewidth]{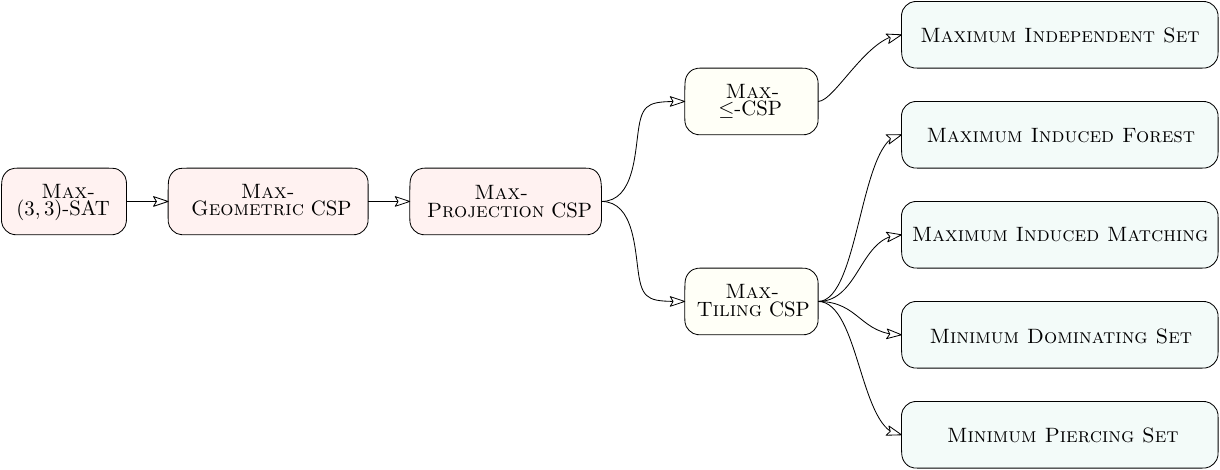}
    \caption{Reduction chain used in the proof of \Cref{thm:main-gcsp} (yellow boxes). The green boxes contain the problems for which we prove hardness in \Cref{thm:application}.}
    \label{fig:reduction-chain}
\end{figure}

\section{Preliminaries}

For a positive integer $n$, we write $[n] = \{1, 2, \dots, n\}$ and $[n]_\circ = \{0,1,\dots,n-1\}$. For a point $x \in \IR^d$, we denote its $i$-th coordinate by $x(i)$. The standard basis vectors of $\IR^d$ are $\dir{e}_1, \dir{e}_2, \dots, \dir{e}_d$. We write $\|x\|_1 = \sum_{i=1}^d |x(i)|$ for the $\ell_1$-norm and $\|x\| = \sqrt{\sum_{i=1}^d x(i)^2}$ for the Euclidean norm.

For integers $k, d > 0$, the \emph{$d$-dimensional grid graph} of side length $k$, denoted $G_k^d$, is the graph with vertex set $[k]^d$ in which two vertices $\ve{a}, \ve{b} \in V(G_k^d)$ are adjacent if and only if $\|\ve{a} - \ve{b}\|_1 = 1$. Throughout, a \emph{unit ball} in $\IR^d$ is an open ball of diameter~$1$. Two unit balls \emph{intersect} if their interiors overlap; equivalently, if their centers are at Euclidean distance strictly less than~$1$. A \emph{unit box} in $\IR^d$ is an open axis-aligned hypercube with side length~$1$. The \emph{intersection graph} of a collection $\cB$ of geometric objects has one vertex per object, with two vertices adjacent whenever the corresponding objects intersect.

An \emph{independent set} in a graph $G$ is a subset of vertices no two of which are adjacent. The \MIS{} problem asks for an independent set of maximum cardinality in a given graph.

An \emph{induced forest} of a graph $G$ is a subset $S \subseteq V(G)$ such that the subgraph of $G$ induced by $S$ is acyclic. The \MIF{} problem asks for an induced forest of maximum cardinality in a given graph.

An \emph{induced matching} in a graph $G$ is a set $M$ of edges such that the subgraph of $G$ induced by the endpoints of~$M$ is a $1$-regular graph (a perfect matching on those vertices). Equivalently, no two edges in $M$ share an endpoint and no edge of $G$ connects endpoints of distinct edges in~$M$. The \MIM{} problem asks for an induced matching of maximum cardinality in a given graph.

A \emph{dominating set} of a graph $G = (V, E)$ is a subset $S \subseteq V$ such that every vertex $v \in V \setminus S$ has at least one neighbor in~$S$. The \MDS{} problem asks for a dominating set of minimum cardinality.

A \emph{piercing set} for a collection $\cB$ of geometric objects is a point set $P \subseteq \IR^d$ such that every object in $\cB$ contains at least one point of $P$. The \MPS{} problem asks for a piercing set of minimum cardinality. A collection of geometric objects $\cB$ is a \emph{cover} of a point set $P$ if each point of $P$ is contained in at least one object of $\cB$. When $\cB$ is a collection of unit balls or unit boxes, one can use the following equivalent \emph{covering} formulation of \MPS{}: given a finite point set $P \subseteq \IR^d$, the goal is to find the minimum number of unit balls or unit boxes whose union covers~$P$.

We study these combinatorial optimization problems on intersection graphs. In geometric settings, the input is a collection $\cB$ of geometric objects (such as unit balls or unit boxes in $\IR^d$), and the goal is to solve each of the aforementioned problems in the intersection graph of $\cB$.

A \emph{polynomial-time approximation scheme} (PTAS) for a maximization (resp.\ minimization) problem is a family of algorithms parameterized by $\eps > 0$ that, for every fixed $\eps > 0$, runs in time polynomial in the input size and returns a solution of value at least $(1-\eps)\OPT$ (resp.\ at most $(1+\eps)\OPT$).

Our hardness results rely on the following hypothesis.

\begin{gapeth}[Dinur~\cite{Dinur16}, Manurangsi and Raghavendra~\cite{ManurangsiR17}]
    There exist constants $\delta, \gamma > 0$ such that there is no $2^{\gamma m}$-time algorithm that, given a $3$-CNF formula $\phi$ on $m$ clauses, can distinguish between the cases where 
    \begin{enumerate}
        \item[(i)] $\phi$ is satisfiable, or
        \item[(ii)] all assignments violate at least $\delta m$ clauses.
    \end{enumerate}
\end{gapeth}

A Boolean formula $\phi$ in CNF is an $(a,b)$-CNF formula if each variable occurs at most $a$ times and each clause has at most $b$ literals. We consider the \textsc{Max-}$(3,3)$-SAT problem, in which the objective is to maximize the number of satisfied clauses in a $(3,3)$-CNF formula. Note that for $(3,3)$-CNF formulas, the number of variables and the number of clauses are within constant factors of each other. Papadimitriou~\cite{Papadimitriou} gives an L-reduction from \textsc{Max-}$3$-SAT to \textsc{Max-}$(3,3)$-SAT, which implies the following:

\begin{lem}[Corollary of~\cite{Papadimitriou}] \label{lem:gap-eth}
There exist constants $\delta, \gamma > 0$ such that, unless Gap-ETH fails, there is no algorithm with running time $2^{\gamma n}$ that, given a $(3,3)$-CNF formula $\phi$ with $n$ variables and $m$ clauses, can distinguish between the following two cases:
    \begin{enumerate}
        \item[(i)] $\phi$ is satisfiable; and
        \item[(ii)] for every assignment at least $\delta m$ clauses of $\phi$ are violated.
    \end{enumerate}
\end{lem}

\section{From (3,3)-SAT to \gcsp} \label{sec:sat-to-gcsp}

In this section, we prove a hardness result for \textsc{Max-}\gcsp{} on $d$-dimensional grids ($d \ge 2$) via an approximation-preserving reduction from \textsc{Max-}$(3,3)$-SAT. To bridge the gap between the combinatorial structure of a $(3,3)$-CNF formula $\phi$ and the geometry of a $d$-dimensional grid in \gcsp, we utilize a strong form of the Cube Wiring Theorem~\cite{sandorPHDThesis} to embed the incidence graph of $\phi$ on a $d$-dimensional grid for~$d \ge 3$.

\begin{thm}[Strong Cube Wiring Theorem~\cite{sandorPHDThesis}]
\label{thm:strong-cube-wiring}
Let $d \ge 3$ and let $G=(A\cup B,E)$ be a bipartite graph on $n$ vertices of maximum degree at most $2d$. Let $n_1 = \Theta\left(n^{1/(d-1)}\right)$ be an integer. Then there exists a constant $c>0$ and injective mappings
\[
\emb_A : A \to [c n_1]^{\,d-1}\times\{1\} \subseteq V(G_{c n_1}^d),
\qquad
\emb_B : B \to [c n_1]^{\,d-1}\times\{c n_1\} \subseteq V(G_{c n_1}^d),
\]
such that the following holds. There exists a collection of $|E|$ internally vertex-disjoint paths $\mathcal{W}{=}\{W_e\}_{e\in E}$ in the $d$-dimensional grid graph $G^d_{c n_1}$ with the property that for every edge $e=\{a,b\}\in E$ with $a\in A$ and $b \in B$, the path $W_e$ connects $\emb_A(a)$ to $\emb_B(b)$. We refer to these paths as wires. Moreover, given $G$ and $d$ the embedding can be constructed in polynomial time.
\end{thm}

For $d=2$, we replace vertex-disjoint wires in \Cref{thm:strong-cube-wiring} with an \textit{edge-disjoint} wiring in a $2$-dimensional grid (allowing crossings, but no shared edges) as in~\cite{sandorPHDThesis}. This is formalized by the following $2$-dimensional grid-wiring lemma.

\begin{lem}[$2$-dimensional Grid Wiring]
\label{lem:2d-wiring}
Let $G=(A\cup B,E)$ be a bipartite graph on $n$ vertices with maximum degree at most $3$ and $|E|=m$.
Then there exists an integer $n_1=\Theta(n+m)$, a constant $c>0$, injective mappings
\[
\emb_A : A \to [c n_1]\times\{1\}\subseteq V(G^2_{c n_1}),
\qquad
\emb_B : B \to [c n_1]\times\{c n_1\}\subseteq V(G^2_{c n_1}),
\]
and a collection $\mathcal{W}=\{W_e\}_{e\in E}$ of \emph{edge-disjoint} grid paths in $G^2_{c n_1}$ such that each $W_e$ connects $\emb_A(a)$ to $\emb_B(b)$ for each edge $e=\{a,b\} \in E$ and every grid vertex belongs to at most three wires; if two wires meet at a grid vertex then they do not share any grid edge (they only cross). Moreover, given $G$, the embedding can be constructed in polynomial time.
\end{lem}

\begin{proof}

Let $A = \{u_1,u_2,\dots, u_{n_A}\}$ and $B = \{v_1,v_2,\dots, v_{n_B}\}$. That is $n = n_A + n_B$. Set $n_1 = n + m$ and $c = 5$. We construct the embedding in the grid $G_{5n_1}^2$. We map the elements of $A$ to the path induced by $[5n_1]\times\{1\}$ and the elements of $B$ to the path induced by $[5n_1] \times \{5n_1\}$ with spacing~3:
\[
\begin{aligned}
\emb_A(u_i) &= (3i, 1) && \text{for } i\in[n_A],\\
\emb_B(v_j) &= (3(n_A+j), 5n_1) && \text{for } j\in[n_B].
\end{aligned}
\]
These mappings are injective since the first coordinates are distinct. Fix an arbitrary ordering $E=\{e_1,e_2,\dots,e_m\}$.
For each $k\in[m]$, we assign a unique row $y_k = k+2$ from the grid, so $3\le y_k\le m+2$. Since the maximum degree of $G$ is at most $3$, for every vertex $w\in A\cup B$ and each incident edge $e$ to~$w$, we assign a distinct local offset $\delta(e,w)\in\{-1,0,1\}$.
For $e_k=\{u_i,v_j\}$ with $u_i \in A$ and $v_j \in B$, for convenience, we define $x_A(e_k) = 3i + \delta(e_k,u_i)$ and $x_B(e_k) = 3(n_A+j) + \delta(e_k,v_j)$. We construct $W_{e_k}$ as the concatenation of the following grid subpaths in the following order.

\begin{itemize}
    \item A (possibly empty) horizontal subpath from $\emb_A(u_i)=(3i,1)$ to $(x_A(e_k),1)$.
    \item A vertical subpath from $(x_A(e_k),1)$ to $(x_A(e_k),y_k)$.
    \item A horizontal subpath from $(x_A(e_k),y_k)$ to $(x_B(e_k),y_k)$.
    \item A vertical subpath from $(x_B(e_k),y_k)$ to $(x_B(e_k),5n_1)$.
    \item A (possibly empty) horizontal subpath from $(x_B(e_k),5n_1)$ to $\emb_B(v_j)=(3(n_A+j),5n_1)$.
\end{itemize}
By construction, it is immediate from the construction that $W_{e_k}$ connects $\emb_A(u_i)$ to $\emb_B(v_j)$. Edge-disjointness holds because all middle rows $y_k$ are distinct, and all columns $\{x_A(e):e\in E\}$ and $\{x_B(e):e\in E\}$ are pairwise distinct (since spacing~$3$ plus distinct local offsets at each vertex), so no two wires share a horizontal or vertical grid edge. Note that if two wires meet, one is vertical and the other horizontal, so they share no grid edge (they only cross). This also implies that every grid vertex that is not an endpoint of a wire belongs to at most two wires (at most one wire continuing horizontally and one wire continuing vertically). Grid vertices that are extremes of a wire belong to at most three wires by hypothesis. The construction of this grid embedding (choosing offsets and five segments per edge) is computable in $O(n_1^2)$ time.
\end{proof}

The main result of this section is the following hardness theorem which gives an approximation-preserving reduction from \textsc{Max-}$(3,3)$-SAT to \textsc{Max-}\gcsp{} on $d$-dimensional grids, for every $d \ge 2$.

\begin{thm}
\label{thm:sat-to-gcsp}
Assuming Gap-ETH, for every fixed integer $d \ge 2$ there exists a constant $\gamma > 0$ such that there is no algorithm with running time $N^{\gamma / \eps^{d-1}}$ that, given any $\eps>0$ and any instance $I = (V, \domain, \constraints)$ of $d$-dimensional \gcsp{} of input size $N = \poly(|\constraints|, |\domain|)$, can distinguish between the following two cases:
\begin{enumerate}
    \item[(i)] there exists an assignment satisfying all constraints in $I$, and
    \item[(ii)] every assignment violates at least $\eps |\constraints|$ constraints.
\end{enumerate}
\end{thm}

\begin{proof}

Fix a positive integer $d \ge 2$. Let $\phi$ be a $(3,3)$-CNF formula. We construct an instance of \gcsp{} whose variables correspond to blocks of a $d$-dimensional grid, and whose domain encodes Boolean values propagated along vertex-disjoint paths (edge-disjoint paths for the case $d = 2$). The key ingredients of the reduction are as follows.
\begin{itemize}
    \item We consider the incidence graph of $\phi$, viewed as a bipartite
    graph between variables and clauses. Using the Strong Cube Wiring Theorem (\Cref{thm:strong-cube-wiring}) for $d \ge 3$ and $2$-dimensional Edge-Disjoint Wiring (\Cref{lem:2d-wiring}) for $d = 2$, we embed this bipartite graph into a $d$-dimensional grid so that each variable and clause
    incidence is represented by a disjoint path (a \emph{wire}).
    \item We partition the grid into small cubes and assign to each cube a
    variable whose value encodes the Boolean values of the wires
    passing through that cube. Unary constraints enforce consistency of literals belonging to
    the same variable and satisfaction of clauses, while binary
    constraints enforce consistency of wire values across adjacent cubes.
\end{itemize}

We show that a satisfiable formula $\phi$ yields fully satisfiable \gcsp{} instances, whereas assignments violating a constant fraction of clauses in $\phi$ correspond to assignments violating a proportional fraction of constraints in the \gcsp{} instance.

\subsection*{Embedding the incidence graph of $\phi$}
\label{subsec:incidence}

Assume that the $(3,3)$-CNF formula $\phi$ has variable set $\var_\phi$ and clause set $\cla_\phi$. We associate with $\phi$ its \emph{incidence graph} $G = (A \cup B, E)$, where $A = \var_\phi$, $B = \cla_\phi$, and $\{\phiv, \phic\} \in E$ if and only if the variable $\phiv \in \var_\phi$ appears in the clause $\phic \in \cla_\phi$.
Since $\phi$ is a $(3,3)$-CNF formula, each variable appears in at most three clauses and each clause contains at most three literals. Consequently, the maximum degree of $G$ is at most $3$.

For any fixed $d \ge 3$, we apply the Strong Cube Wiring Theorem (\Cref{thm:strong-cube-wiring}) and for $d = 2$, we apply the $2$-dimensional Edge-Disjoint Wiring (\Cref{lem:2d-wiring}) to embed $G$ into a $d$-dimensional grid.
See \Cref{fig:s-partition}. Let $n_1 = \Theta\bigl(n^{1/(d-1)}\bigr)$ be the grid side length parameter guaranteed by \Cref{thm:strong-cube-wiring} (or \Cref{lem:2d-wiring}).
Moreover, it provides injective mappings
\[
\emb_A : A \to [\gridconst n_1]^{\,d-1} \times \{1\}
\qquad \text{and} \qquad
\emb_B : B \to [\gridconst n_1]^{\,d-1} \times \{\gridconst n_1\},
\]
for a constant $\gridconst > 0$, together with a collection of vertex-disjoint paths (edge-disjoint for $d = 2$) $\mathcal{W} = \{W_e\}_{e \in E}$ in the grid graph $G^d_{\gridconst n_1}$.

Note that the size of the grid is $\Theta\bigl(n^{d/(d-1)}\bigr)$. Each path $\wire_e \in \wires$ connects $\emb_A(\phiv)$ to $\emb_B(\phic)$ for the corresponding incidence $e = \{\phiv, \phic\} \in E$. Recall that we refer to these paths as \emph{wires}. Intuitively, each wire is responsible for transmitting the Boolean value assigned to a literal from its variable vertex to the corresponding clause vertex.

\subsection*{Constructing an instance of \gcsp}

\begin{figure}
    \centering
    \includegraphics[width=0.5\linewidth]{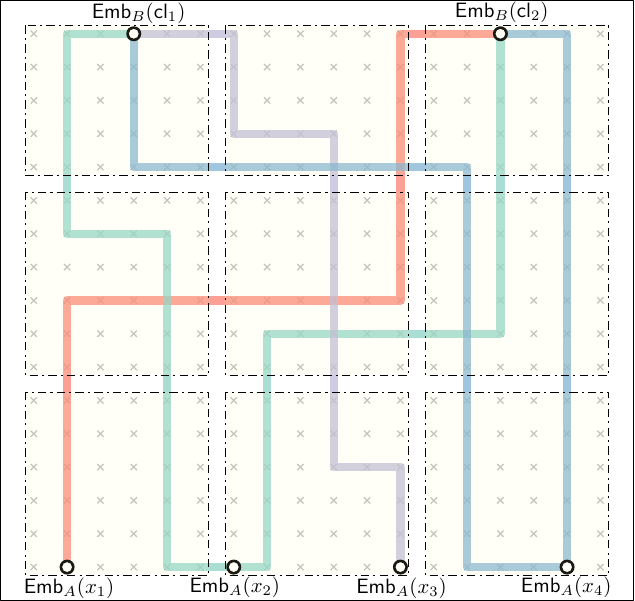}
    \caption{An $s$-partition of a $2$-dimensional embedding of the incidence graph of the formula $\phi = (\bar{x}_2 \lor x_3 \lor x_4) \land (\bar{x}_1 \lor x_2 \lor \bar{x}_4)$ for $s = 6$. The wires corresponding to the same variable share the same color.}
    \label{fig:s-partition}
\end{figure}

We first define the variables and domain of the \gcsp{} instance. Let $s$ be a positive integer parameter. We define an $s$-\emph{partition} of $G^d_{\gridconst n_1}$ by grouping vertices
into axis-aligned $d$-dimensional cubes of side length $s$. Formally, define a mapping $\psi : V(G^d_{\gridconst n_1}) \to V(G^d_{n_2})$, by
\[
\psi(x_1, \dots, x_d)
=
\bigl(\lceil x_1 / s \rceil, \dots, \lceil x_d / s \rceil\bigr),
\]
where $n_2 = \left\lceil \frac{\gridconst n_1}{s} \right\rceil$. For each $v \in V(G^d_{n_2})$, the preimage $\psi^{-1}(v)$ induces a connected subgraph of $G^d_{\gridconst n_1}$ isomorphic to a $d$-dimensional grid of side length at most $s$.
We refer to such a subgraph as a \emph{small cube} and we denote the set of all small cubes by $\Sigma$.
The variables $V$ of the \gcsp{} instance correspond to the small cubes in $\Sigma$, i.e., $V = \{v_\sigma : \sigma \in \Sigma\}$. Naturally, we define the subgraph isomorphism  $\loc : V \rightarrow V(G_{n_2}^d)$ by $\loc(v_\sigma) := \psi(z)$ for an arbitrary vertex~$z$ of~$\sigma$. We say that two small cubes $\sigma$ and $\sigma'$ are adjacent if $\loc(v_\sigma)$ and $\loc(v_{\sigma'})$ are adjacent in $G_{n_2}^d$. See \Cref{fig:s-partition}.


Recall that the collection $\mathcal{W}$ consists of internally vertex-disjoint paths in $G^d_{\gridconst n_1}$ (edge-disjoint paths for $d=2$).
For a small cube $\sigma \in \Sigma$, let $\mathcal{W}_\sigma \subseteq \mathcal{W}$ denote the set of wires that intersect $\sigma$. We bound the number of such wires as follows.
Each wire enters or exits a small cube through one of its $2d$ facets, each facet containing at most $s^{d-1}$ vertices.
Since the wires are internally vertex-disjoint, at most three wires can pass through each boundary vertex (for $d > 2$, only when it is an endpoint of a wire; otherwise, the bound is smaller, and for $d=2$, \Cref{lem:2d-wiring} ensures this property also holds for the edge-disjoint paths). Consequently, $|\mathcal{W}_\sigma| \le 6d \cdot s^{d-1}$.

As a shorthand, we define $\wirebound := 6d \cdot s^{d-1}$. The domain of the \gcsp{} instance is defined as $\domain := \{0,1\}^{\wirebound}$.
Intuitively, each coordinate of a domain element corresponds to the Boolean value carried by one of the wires passing through the relevant small cube.
For each small cube $\sigma$, we fix an arbitrary injective mapping $f_\sigma : \mathcal{W}_\sigma \to [\wirebound]$, which assigns to each wire intersecting $\sigma$ a unique coordinate of the domain vector.

Finally, we define the set of constraints $\constraints$ of the \gcsp{} instance. The constraints are of three types:

\subparagraph{(i) Variable consistency constraints.}
Let $\sigma$ be a small cube that contains the embedded vertex $\emb_A(\phiv)$ corresponding to a variable $\phiv \in \var_\phi$. Let $L(\phiv)$ denote the set of literals of $\phi$ in which $\phiv$ participates.
For the variable $v_\sigma \in V$, we introduce a unary constraint requiring that all wires corresponding to occurrences of $\phiv$ encode a consistent Boolean~assignment.

Formally, the allowed domain values at $v_\sigma \in V$ are those vectors $\ve{x} \in \{0,1\}^\wirebound$ such that for all literals $l, l' \in L(\phiv)$ with corresponding wires $\wire$ and $\wire'$,
\[
\ve{x}(f_\sigma(\wire)) =
\begin{cases}
\ve{x}(f_\sigma(\wire')) & \text{if $l$ and $l'$ have the same sign},\\[4pt]
1 - \ve{x}(f_\sigma(\wire')) & \text{otherwise}.
\end{cases}
\]

\subparagraph{(ii) Clause satisfaction constraints.}
Let $\sigma$ be a small cube that contains the embedded vertex $\emb_B(\phic)$ corresponding to a clause $\phic \in \cla_\phi$.
Let $L(\phic) = \{l_1, l_2, l_3\}$ be the (at most three) literals appearing in the clause $\phic$. For the variable $v_\sigma$, we introduce a unary constraint requiring that at least one of these literals evaluates to true.

Formally, the allowed domain values at $v_\sigma \in V$ are those vectors $\ve{x} \in \{0,1\}^\wirebound$ such that if $\wire_1$, $\wire_2$, and $\wire_3$ are the corresponding wires of the literals $l_1$, $l_2$, and $l_3$ respectively, then
\[ \bigl( \ve x(f_\sigma(\wire_1)), \ve x(f_\sigma(\wire_2)), \ve x(f_\sigma(\wire_3)) \bigr) \neq (0,0,0). \]

\subparagraph{(iii) Wire consistency constraints.}
Finally, consider two adjacent small cubes $\sigma, \sigma' \in \Sigma$. For every wire $\wire \in \mathcal{W}_\sigma \cap \mathcal{W}_{\sigma'}$, the Boolean value carried by $\wire$ must be consistent across both cubes.
We introduce a binary constraint between $v_\sigma$ and $v_{\sigma'}$ requiring equality on all shared wire coordinates.

Formally, the allowed pairs $(\ve{x}, \ve{y}) \in \{0,1\}^\wirebound \times \{0,1\}^\wirebound$ at the binary constraint between $v_\sigma$ and $v_{\sigma'}$ satisfy
\[ \ve x(f_\sigma(\wire)) = \ve y(f_{\sigma'}(\wire))
\qquad \text{for all }
\wire \in \mathcal{W}_\sigma \cap \mathcal{W}_{\sigma'}. \]
This completes the definition of the constraint set $\constraints$.

\subsection*{Exact Equivalence and Error Analysis}

We show that if the Boolean formula $\phi$ is satisfiable, then the constructed \gcsp{} instance admits an assignment that satisfies all constraints. Let $\alpha : \var_\phi \to \{0,1\}$ be a solution of $\phi$. We define an assignment $\mu : V \to \{0,1\}^\wirebound$ to the \gcsp{} variables as follows. Each literal $l$ of $\phi$ corresponds to a unique wire $\wire_l \in \mathcal{W}$ in the grid embedding. For each small cube $\sigma$ and each wire $\wire_l \in \mathcal{W}_\sigma$, where $l$ is a literal over variable $\phiv$, we set
\[
\mu(v_\sigma)\bigl(f_\sigma(\wire_l)\bigr) =
\begin{cases}
\alpha(\phiv) & \text{if $l$ is a positive occurrence of the variable $\phiv$},\\
1 - \alpha(\phiv) & \text{if $l$ is a negated occurrence of the variable $\phiv$}.
\end{cases}
\]
All remaining coordinates of $\mu(v_\sigma)$ (corresponding to wires not intersecting $\sigma$) may be~assigned arbitrarily. By construction, $\mu$ satisfies all constraints of the \gcsp{} instance.

Conversely, consider an arbitrary assignment $\lambda : V \to \{0,1\}^\wirebound$ to the \gcsp{} instance. We define a Boolean assignment $\alpha : \var_\phi \to \{0,1\}$ for the formula $\phi$ as follows. For each variable $\phiv \in \var_\phi$, let $\sigma(\phiv)$ be the unique small cube containing the embedded vertex $\emb_A(\phiv)$. If the variable-consistency constraint at $\sigma(\phiv)$ is satisfied by $\lambda$, then all wires corresponding to occurrences of $\phiv$ encode a consistent value; we set $\alpha(\phiv)$ to that value. Otherwise, we assign $\alpha(\phiv)$ arbitrarily. Again by construction, if $\lambda$ is a solution of the \gcsp{} instance, then $\alpha$ is a solution of~$\phi$.

\begin{clm} \label{clm:clause-violation}
Each violated constraint under the assignment $\lambda$ for the \gcsp{} instance can cause at most~$6\wirebound$ violated clauses in the corresponding assignment $\alpha$ for the formula $\phi$.
\end{clm}

\begin{proof}
    Any violated constraint, whether unary or binary, involves at most two small cubes (a pair of adjacent small cubes in case of binary constraint) whose value assignment may corrupt the Boolean values carried by all wires intersecting the cube(s). Each such cube intersects at most $\wirebound$ wires, and each wire corresponds to a literal occurrence that appears in at most three clauses. Therefore, each violated constraint can affect at most $6\wirebound$ clauses of $\phi$.
\end{proof}

The number of variables, hence the number of constraints up to a constant factor, is
\[
|\constraints|=\Theta(n_2^d)=\Theta\left(\left(\frac{n_1}{s}\right)^d\right).
\]
Moreover, for $(3,3)$-formulas we have $m=\Theta(n)$ and $|E|=\Theta(n)$. Suppose that at most $\eps|\constraints|$ constraints are violated by the arbitrary \gcsp{} assignment $\lambda$. Then by \Cref{clm:clause-violation} the number of clauses violated in the induced Boolean assignment $\alpha$ is at most 
\[
6\wirebound\cdot \eps|\constraints|
= \Theta(s^{d-1})\cdot \eps\cdot \Theta\left(\left(\frac{n_1}{s}\right)^d\right) = \eps \cdot \Theta\left(\frac{n_1^d}{s}\right) = \eps\cdot \Theta\left(\frac{n^{d/(d-1)}}{s}\right).
\]

Assume that Gap-ETH holds, and let $\gamma$ and $\delta$ be the two constants from \Cref{lem:gap-eth}. By \Cref{clm:clause-violation}, setting $\eps := \delta \, s \, c_d /{n^{1/(d-1)}}$ for an appropriate constant $c_d > 0$ yields that if every assignment of $\phi$ violates at least $\delta m$ clauses, every assignment to the \gcsp{} instance violates at least $\eps|\constraints|$ constraints.

The constructed \gcsp{} instance has $|\constraints| = \Theta\left(n^{d/(d-1)} / s^d\right)$ and domain size $|D| = 2^\wirebound = 2^{\Theta(s^{d-1})}$ and it is obtained from $\phi$ in polynomial time in the instance size $N = \poly(|\constraints|, |\domain|)$. Suppose there exists an algorithm that, for every $\gamma_0>0$, distinguishes satisfiable \gcsp{} instances from those in which every assignment violates at least $\eps|\constraints|$ constraints in time $N^{\gamma_0/\eps^{d-1}}$. Applying this algorithm to the instance produced by the reduction for values of $s = \Omega(\log^{1/(d-1)}n)$ yields an algorithm that distinguishes satisfiable $(3,3)$-CNF formulas from those in which every assignment violates at least $\delta m$ clauses in~time

\begin{align*}
\poly\bigl(|\constraints|, |D|\bigr)^{\Theta(\gamma_0 n/s^{d-1})}
&\le
\left(
2^{\Theta(s^{d-1})}\cdot \frac{n^{d/(d-1)}}{s^d}
\right)^{\Theta(\gamma_0 n/s^{d-1})}
\\[6pt]
&=
2^{\Theta(s^{d-1})\cdot \Theta(\gamma_0 n/s^{d-1})}
\cdot
\left(\frac{n^{d/(d-1)}}{s^d}\right)^{\Theta(\gamma_0 n/s^{d-1})}
\\[6pt]
&=
2^{\gamma_0\Theta(n)}
\cdot
2^{\Theta\left(\frac{\gamma_0 n\log(n/s)}{s^{d-1}}\right)}
\\[6pt]
&=
2^{\gamma_0 c'_d n},
\end{align*}
for some constant $c'_d > 0$. Choosing $\gamma_0>0$ sufficiently small contradicts \Cref{lem:gap-eth}.
\end{proof}

\section{\gcsp{} with (In)equality} \label{sec:tcsp}

We continue the reduction chain to $d$-dimensional \lecsp{} and \tcsp{} (thereby proving \Cref{thm:main-gcsp}) by showing that the reductions of Proposition~2.18 and Proposition~2.19 of~\cite[Section 2.3]{MarxS14} are approximation preserving. See \Cref{fig:reduction-chain}.

\subsection{From \gcsp{} to \pcsp{}}

As an intermediate step, we first show that the hardness of approximation for \textsc{Max-}\gcsp{} transfers to \textsc{Max-}\pcsp{}. An instance $I = (V, \domain, \constraints)$ of \gcsp{} is also an instance of \pcsp{} if every binary constraint $C_{\ve{a},\ve{b}} \in \constraints$ is of the form $\{(x,y) \in D \times D : y = p(x)\}$ (projection from $\ve{a}$ to $\ve{b}$) or of the form $\{(x,y) \in D \times D : x = p(y)\}$ (projection from $\ve{b}$ to $\ve{a}$) for some function $p:\domain \to \domain$.

Again, in the \textsc{Max-}\pcsp, the goal is to maximize the number of satisfied constraints. We use the same reduction as in~\cite[Proposition 2.18]{MarxS14} where a $d$-dimensional \gcsp{} instance $I = (V,\domain,\constraints)$ is reduced to an instance $I' = (V', \domain^2, \constraints')$ of $d$-dimensional \pcsp{} while preserving satisfiability. We show that their reduction is approximation preserving.

\begin{prop}
    \label{prop:gcsp-to-pgcsp}
    Assuming Gap-ETH, for every integer $d \ge 2$, there exists a constant $\gamma > 0$ such that there is no algorithm with running time $N^{\gamma / \eps^{d-1}}$ that, given any $\eps>0$ and instance $I = (V,\domain,\constraints)$ of $d$-dimensional \pcsp{} of input size $N = \poly(|\constraints|,|\domain|)$, can distinguish between the following two cases:
    \begin{enumerate}
        \item[(i)] there exists an assignment satisfying all constraints, and
        \item[(ii)] every assignment violates at least $\eps|\constraints|$
        constraints.
    \end{enumerate}
\end{prop}

\begin{proof}
Let $I=(V,\domain,\constraints)$ be an arbitrary instance of $d$-dimensional \gcsp{} with subgraph isomorphism $\loc : V \rightarrow V(G_{k}^d)$. By Proposition~2.18 of Marx and Sidiropoulos~\cite{MarxS14}, there exists a polynomial-time reduction that creates an equivalent instance $I'=(V',\domain^2,\constraints')$ of $d$-dimensional \pcsp{} from $I$ with subgraph isomorphism $\loc' : V' \rightarrow V(G_{2k}^d)$ such that $I$ is satisfiable if and only if $I'$ is satisfiable. We do not repeat the reduction here. Instead, we only state the properties of the reduction that are relevant for the approximation analysis. Our argument is based on a charging scheme.

In this reduction, each unary constraint of $I$ corresponds to exactly one unary constraint in $I'$ and each binary constraint of $I$ corresponds to three constraints in $I'$: one unary constraint and two projection constraints. Formally, fix an element $z_0 \in D$. Each variable $\ve{a} \in V$ is mapped to variable $\ve{a}' \in V'$ such that $\loc'(\ve{a}') = 2\,\loc(\ve{a})$, and the unary constraint $C_\ve{a} \subseteq \domain$ of the variable $\ve{a}$ is mapped to a unary constraint $\{(x,z_0) : x \in C_\ve{a}\} \subseteq \domain^2$ on $\ve{a}'$. A binary constraint $C_{\ve{a}, \ve{b}} \subseteq \domain^2$ is enforced by

\begin{itemize}
    \item Imposing $C_{\ve a,\ve b}$ as a unary constraint on the variable $\ve{\mathsf{mid}_{\ve a,\ve b}}$ with $\loc'(\ve{\mathsf{mid}_{\ve a,\ve b}}) = \loc(\ve{a}) + \loc(\ve{b})$,
    \item Setting $\{((x,z_0), (x,y)) : x,y \in \domain\}$ as a binary constraint between $\ve{a}'$ and $\ve{\mathsf{mid}_{\ve a,\ve b}}$, and
    \item Setting $\{((x,y),(y,z_0)) : x,y \in \domain\}$ as a binary constraint between $\ve{\mathsf{mid}_{\ve a,\ve b}}$ and $\ve{b'}$.
\end{itemize}

The variables and constraints introduced above are called \emph{important} in~\cite{MarxS14}. We charge each important constraint of $I'$ to the corresponding constraint in $I$. In addition, the construction introduces \emph{unimportant} constraints. Each such constraint lies within a box of side length~$2$ centered at an important variable in the primal graph of $I'$. We charge each unimportant constraint to a nearest important constraint in the primal graph, and hence equivalently to the corresponding constraint in $I$. With this charging scheme, each constraint in $I$ is charged at most $O(3^d)$ times. Therefore, there exists an explicit constant $c_d = O(3^d)$ such that $|\constraints'| \le c_d \, |\constraints|$.

Assume that every assignment for $I$ violates at least $\eps|\constraints|$ constraints. Fix an arbitrary assignment $\bar f$ to $I'$ and let $f$ be the induced assignment to $I$ from~$\bar f$. For every violated constraint $c\in\constraints$ of $f$, the construction guarantees that $\bar f$ violates at least one constraint corresponding to $c$. Hence $\bar f$ violates at least $\eps|\constraints|$ constraints of $I'$ and therefore at least $(\eps/c_d)|\constraints'|$ constraints.

The reduction increases the input size only polynomially. If $N'=\poly(|\constraints'|,|\domain^2|)$ denotes the input size of $I'$, then $N'=\poly(N)$.
For every $\gamma_0 > 0$, suppose there exists an algorithm for $d$-dimensional
\pcsp{} that can distinguish satisfiable instances from those where every
assignment violates a $\delta$-fraction of constraints, running in time
${N'}^{\gamma_0/\delta^{d-1}}$ on instances of input size $N'$. Applying it to
$I'$ with $\delta = \eps/c_d$ yields a running time of
\[
  {N'}^{\gamma_0 c_d^{d-1}/\eps^{d-1}}
  = \poly(N)^{\gamma_0 c_d^{d-1}/\eps^{d-1}}
  = N^{\gamma_0c'_d/\eps^{d-1}},
\]
for a suitable constant $c'_d > 0$. This gives an algorithm that distinguishes satisfiable
instances of $d$-dimensional \gcsp{} from those where every assignment violates
at least an $\eps$-fraction of constraints in time $N^{\gamma/\eps^{d-1}}$,
contradicting \Cref{thm:sat-to-gcsp} under Gap-ETH (for small enough~$\gamma_0$).
\end{proof}

\subsection{From \pcsp{} to \texorpdfstring{$\le$}{ <=}-CSP}

We now reduce from \pcsp{} to \lecsp{}. Recall that an instance of \lecsp{} can contain arbitrary unary constraints, but the binary constraints $C_{\ve{a},\ve{b}}$ are of the form $\{(x,y) \in \domain\times \domain : x(i) \le y(i)\}$ where $\ve{a}$ and $\ve{b}$ are adjacent variables with $\loc(\ve{b}) = \loc(\ve{a})+\dir{e}_i$ for some $1\le i \le d$ and $D = [\Lambda]^d$ is the domain for some positive integer $\Lambda$.

\begin{thm}
    \label{thm:pgcsp-to-lecsp}
    Assuming Gap-ETH, for every integer $d \ge 2$ there exists a constant $\gamma > 0$ such that there is no algorithm with running time $N^{\gamma / \eps^{d-1}}$ that, given any $\eps>0$ and instance $I = (V,D, \constraints)$ of~$d$-dimensional \lecsp{} of input size $N = \poly(|\constraints|,|D|)$, can distinguish between the following two cases:
    \begin{enumerate}
        \item[(i)] there exists an assignment satisfying all constraints, and
        \item[(ii)] every assignment violates at least an $\eps$-fraction of the constraints.
    \end{enumerate}
\end{thm}

\begin{proof}
    Let $I=(V,\domain,\constraints)$ be an arbitrary instance of a $d$-dimensional \pcsp{} where the primal graph of $I$ is isomorphic to an induced subgraph of the grid $G_{k}^d$ for some positive integer $k$.
    Proposition~2.19 of~\cite{MarxS14} gives a polynomial-time reduction that produces an instance $I'=(V',\domain',\constraints')$ of \lecsp{} with subgraph isomorphism $\loc' : V' \rightarrow V(G_{ck}^d)$ between the primal graph of $I'$ and an induced subgraph of the grid $G_{ck}^d$ for some positive constant $c$, and the domain is $\domain'=[\Lambda]^d$ for $\Lambda = 2|\domain|+1$.
    
    Following the reasoning of~\cite{MarxS14}, for each solution $f$ of $I$, there exists a corresponding solution $\bar f$ of $I'$, and vice versa. That is, $I$ is satisfiable if and only if $I'$ is satisfiable.
    
    For each variable $\ve{a} \in V$, the construction introduces a cycle gadget consisting of exactly $12$ variables, $\ve{a}'_1, \ve{a}'_2, \dots, \ve{a}'_{12}$, all located within a constant-size neighborhood of the grid vertex $\ve{a}^\star = (5a_1, 5a_2, 2a_3, \dots, 2a_d) \in V(G_{ck}^d)$. That is $\|\loc'(\ve{a}'_i)-\ve{a}^\star\|_1 \le \zeta$, for each $1\le i\le 12$ and some constant $\zeta > 0$. Unary constraints are imposed on each of these $12$ variables, and binary inequality constraints are added along edges in the cycle. Each projection constraint in $I$ between adjacent variables $\ve{a}$ and $\ve{b}$ with $\loc(b) = \loc(a) + \dir{e}_i$ for some $1 \le i \le d$, is replaced by two new variables and four binary inequality constraints between the corresponding cycle gadgets in $I'$.

    Similar to the argument in the proof of \Cref{prop:gcsp-to-pgcsp}, the construction is local in the sense that each variable $\ve{a}$ in $I$, together with the constraints incident to it, is mapped to variables and constraints within a constant-size neighborhood in the primal graph of the instance $I'$. Therefore, there exists a constant $c_d = 2^{O(d)}$ such that $|\constraints| \le c_d \, |\constraints'|$. By the correspondence of assignments between the two instances, if an assignment $f$ violates a constraint in $I$, there is at least one corresponding constraint in $I'$ that is violated by the assignment $\bar f$. A charging argument similar to the one used in the proof of \Cref{prop:gcsp-to-pgcsp} concludes the proof.
\end{proof}

\subsection{From \pcsp{} to \tcsp{}}

Recall that the definition of \tcsp{} is identical to that of \lecsp{}, except that the inequality constraints are replaced the inequality ($\le$) constraints by equalities ($=$). As such, one can make the same replacement, changing $\le$ to $=$, in the proof of \Cref{thm:pgcsp-to-lecsp} and all the arguments stay valid. We state the analogue of \Cref{thm:pgcsp-to-lecsp} for \tcsp{}.\footnote{One can construct a \tcsp{} instance from $I$ in which each variable of $I$ is replaced by a cycle with $8$ variables, and each binary projective constraint of $I$ is replaced by a new variable together with two equality constraints. This would yield a construction by a constant factor smaller instance size. However, for our analysis, improving the constant factor does not strengthen the hardness result. Therefore, we proceed by borrowing the construction from~\cite{MarxS14}.}

\begin{thm}
    \label{thm:pgcsp-to-tcsp}
    Assuming Gap-ETH, for every integer $d \ge 2$ there exists a constant $\gamma > 0$ such that there is no algorithm with running time $N^{\gamma / \eps^{d-1}}$ that, given any $\eps>0$ and instance $I = (V,D,\constraints)$ of $d$-dimensional \tcsp{} of input size $N = \poly(|\constraints|,|D|)$, can distinguish between the following two cases:
    \begin{enumerate}
        \item[(i)] there exists an assignment satisfying all constraints, and
        \item[(ii)] every assignment violates at least an $\eps$-fraction of the constraints.
    \end{enumerate}
\end{thm}

\section{Hardness of Approximation via \tcsp{} and \texorpdfstring{$\le$}{<=}-CSP} \label{sec:lower_bound}

As applications of \Cref{thm:main-gcsp}, we derive running-time lower bounds for approximation schemes for \MIS{}, \MIF{}, \MIM{}, \MDS{}, and \MPS{} on intersection graphs of unit balls (and unit cubes) in~$\R^d$. Our framework is based on reductions from either \textsc{Max-}\tcsp{} or \textsc{Max-}\lecsp{}. For \MIS{} and \MDS{}, we can directly employ existing gadget constructions~\cite{Marx06MDS,MarxS14}, whereas \MPS{} requires the more involved gadgetry recently developed for the $k$-center problem by Blank~\etal~\cite{geert}. For \MIF{} and \MIM{}, we present a new reduction inspired by the reduction used for \MIS{} in \cite{Marx07}. As a matching upper bound, in \Cref{sec:algorithms}, we discuss shifting based algorithms for these problems.

\subsection{\MIS{} Problem}

\begin{thm} \label{thm:mis}
    Let $d \ge 2$ be an integer. There is a real $\gamma > 0$ such that if the \MIS{} on intersection graphs of $n$ unit balls in $\mathbb{R}^d$ admits a PTAS with running time $n^{\gamma/\eps^{d-1}}$ then Gap-ETH~fails.
\end{thm}

\begin{proof}
    Let $d \ge 2$ be a fixed integer. Let $I=(V,\domain,\constraints)$ be an instance of $d$-dimensional \lecsp{} with domain $D = [\Lambda]^d$ for a positive integer $\Lambda$ and with subgraph isomorphism $\loc : V \rightarrow V(G_k^d)$ that maps the variables of $I$ to vertices of the grid $G_k^d$ for a positive integer $k$. Let $\totalconstraints$ be the total number of constraints (unary and binary). Following the construction in~\cite[Theorem~3.1]{MarxS14}, we build a set of unit balls in $\R^d$ from~$I$ as an instance of \MIS{}.
    In this subsection, we denote the size of \MIS{} of a set $\cB$ of unit balls in $\R^d$ by~$\OPT(\mathcal{B})$.
    
    \subparagraph{Construction.} Let $\rho = 1/ (3\cdot d\cdot \Lambda^2)$. For every variable $\ve a\in V$ and every $\ve x\in C_{\ve a}\subseteq [\Lambda]^d$, create a ball $B(\ve a,\ve x)$ centered at $\loc(\ve a)+\rho\,\ve x \in \R^d$. For each $\ve a\in V$ let $\mathcal{B}_{\ve a} = \{B(\ve a,\ve x):\ve x\in C_{\ve a}\}$. Let $\cB$ be the set of all created balls. That is $\mathcal{B} = \bigcup_{\ve{a}\in V} \mathcal{B}_\ve{a}$. The size of this instance is $|\mathcal{B}|=\sum_{\ve a\in V}|C_{\ve a}|\ \le\ |V| \Lambda^d = O(\poly(|\constraints|,|\domain|)$.

    Note that the balls in $\mathcal{B}_\ve{a}$ create a clique in the intersection graph of $\cB$. If $\ve{a} \in V$ and $\ve x, \ve y \in C_{\ve a}$, then by the choice of $\rho$, $B(\ve{a},\ve{x})$ and $B(\ve{a},\ve{y})$ intersect. Moreover cliques of two non-adjacent variable are disjoint. If $\ve{a} \in V$ and $\ve{b} \in V$ are two distinct variables that are not adjacent, then $\|\ve a-\ve b\|\ge \sqrt{2}$ and by the choice of $\rho$ centers of balls in $\mathcal{B}_\ve{a}$ and  $\mathcal{B}_\ve{b}$ are at distance at least $1$ from each other. The following claim from \cite{MarxS14} demonstrates that the intersection pattern of adjacent variables encodes the binary constraints.

    \begin{clm}[\cite{MarxS14}]\label{clm:mis-disjointness}
         Let $\ve{a},\ve{b}\in V$ such that $\loc(\ve b)=\loc(\ve a)+\dir{e}_i$ for some $i\in[d]$. Then for $\ve{x} \in C_\ve{a}$ and $\ve{y}\in C_\ve{b}$, the unit balls $B(\ve a,\ve x)$ and $B(\ve b,\ve y)$ are disjoint if and only if $x(i)\le y(i)$.
    \end{clm}
    
    \subparagraph{Exact Equivalence and Error Analysis.} For the exact equivalence, Marx and Sidiropoulos~\cite{MarxS14} prove that $\OPT(\mathcal{B}) = |V|$ if and only if $I$ is satisfiable. Assume now that every assignment of $I$ violates at least a $\delta$-fraction of the constraints. Let $S$ be an arbitrary independent set of balls. The set $S$ contains balls from $|S|$ distinct variables. Let $U \subseteq V$ be those variables. The set $S$ induces a partial assignment $f$ on $U$ by $f(\ve a) = \ve{x}$ whenever $B(\ve a,\ve x)\in S$. We extend $f$ to a total assignment $\bar f$ on all of $V$ arbitrarily, choosing for each $\ve{a} \in V \setminus U$ some value from $D$. Then all unary constraints for variables in $U$ are satisfied by construction. Every binary constraint whose two endpoints both lie in $U$ is satisfied. If two adjacent variables $\ve{a}$ and $\ve{b}$ are both in $U$, then $B(\ve a,\bar f(\ve a))$ and $B(\ve b,\bar f(\ve b))$ are both in the independent set $S$. Therefore they are disjoint. If $\bar f(\ve a) = (x_1, x_2,\dots, x_d) \in C_\ve{a}$ and $\bar f(\ve b) = (y_1, y_2,\dots, y_d)\in C_\ve{b}$ and $a$ and $b$ are adjacent with $\loc(b) = \loc(a) +\dir{e}_i$ (for some $i \in [d]$), then \Cref{clm:mis-disjointness} implies $x_i \le y_i$.

    Thus, the only constraints that may be violated are those incident to $V\setminus U$. Each variable participates in at most $2d$ binary constraints, and in exactly one unary constraint. Therefore the number of constraints that can possibly be violated is at most $(2d+1)(|V|-|S|)$. Since every assignment violates at least $\delta \, \totalconstraints$ constraints, we obtain $(2d+1)(|V|-|S|) \ge \delta \, \totalconstraints$. Therefore $|S| \le |V|-\frac{\delta \cdot \totalconstraints}{2d+1}$. In particular, as $\totalconstraints \ge |V|$, the size of maximum independent set of $\cB$ is at most
    \[
         \Bigl(1-\frac{\delta}{2d+1}\Bigr)|V|.
    \]

    \subparagraph{Running Time.} Assume toward a contradiction that there is a PTAS for \MIS{} of $n$ unit balls in $\R^d$ running in time $n^{\gamma_0/\eps^{\,d-1}}$ for every constant $\gamma_0>0$. We run this PTAS with $\eps = \frac{\delta}{4(2d+1)}$ on the instance $\cB$. If $I$ is satisfiable, the algorithm returns an independent set of size at least $(1-\eps)|V| > \bigl(1-\frac{3\delta}{4(2d+1)}\bigr)|V|$. If every assignment of $I$ violates a $\delta$-fraction of constraints then every independent set of $\cB$ has size at most $\bigl(1-\frac{\delta}{2d+1}\bigr)|V| < \bigl(1-\frac{3\delta}{4(2d+1)}\bigr)|V|$. Hence an algorithm with running time
    \[
    |\mathcal{B}|^{\gamma_0/\eps^{d-1}} = 
    |\mathcal{B}|^{O(1)\gamma_0/\delta^{d-1}} =
    \poly(k,\Lambda)^{\gamma_0 /\delta^{d-1}},
    \]
    can distinguish the satisfiable instances of \lecsp{} from the instances where every assignment violates $\delta$-fraction of the constraints. This contradicts \Cref{thm:pgcsp-to-lecsp} (for a suitable constant $\gamma_0$), and therefore such a PTAS cannot exist unless Gap-ETH fails.
\end{proof}

Following the construction of \cite[Theorem~3.2]{MarxS14} for unit cubes and a similar reasoning as in the proof of \Cref{thm:mis}, we obtain the analogous result for unit cubes in $\R^d$. We omit the details.

\begin{cor} \label{thm:mis-cube}
    Let $d \ge 2$ be an integer. There is a real $\gamma > 0$ such that if the \MIS{} on intersection graphs of $n$ unit cubes in $\mathbb{R}^d$ admits a PTAS with running time $n^{\gamma/\eps^{d-1}}$ then Gap-ETH~fails.
\end{cor}

\subsection{\MIF{} Problem} \label{appendix:lower_bound_mif}

Next, inspired by the construction of~\cite{Marx07} for \MIS{} on unit disks via \mati{} in the plane, we develop a reduction from \tcsp{} to \MIF{} on intersection graphs of unit balls in $\R^d$, for any fixed $d \ge 2$. Moreover, for $d\ge3$, the same reduction extends to intersection graphs of unit cubes in~$\R^d$.

\begin{thm}\label{thm:mif-ball}
    Let $d \ge 2$ be an integer. There is a real $\gamma > 0$ such that if \MIF{} on intersection graphs of $n$ unit balls in $\R^d$ admits a PTAS with running time $n^{\gamma/\eps^{d-1}}$, then Gap-ETH fails.
\end{thm}

\begin{proof}
    We adapt the gadget construction of Marx~\cite{Marx07} given for \MIS{} on unit disks in the plane to the \MIF{} problem on unit balls in~$\R^d$, and reduce from \tcsp{} (\Cref{thm:pgcsp-to-tcsp}). Marx~\cite{Marx07} reduces \mati{} to \MIS{} on unit disks by representing each variable as a cyclic arrangement of cliques around a square, with per-value perturbations chosen so that independent selections around the cycle are forced to encode a single common value (\emph{independent selection}). Adjacent gadgets are tied by two connection cliques, one enforcing a $\le$-relation and the other a $\ge$-relation on the shared coordinate, so that any independent selection makes the two gadgets agree on that coordinate (\emph{coordinate consistency}). See \Cref{fig:marx-gadget}. We refer to~\cite{Marx07} for the full accounting; our reduction reproduces these two mechanisms in the induced-forest setting, where neither is enforced automatically.

    \begin{figure}
        \centering
        \includegraphics[width=0.75\linewidth]{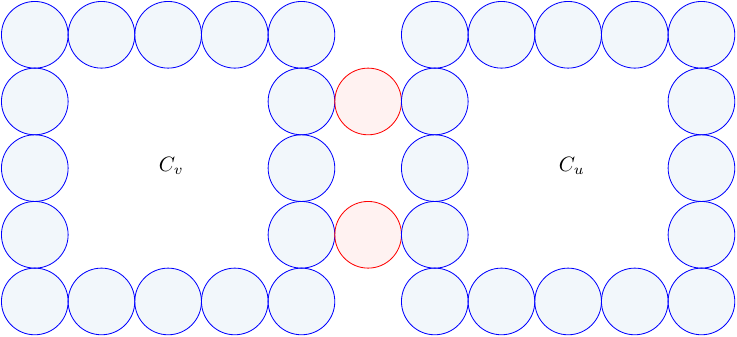}
        \caption{Marx's gadget~\cite{Marx07} for \MIS{} of unit disks via \mati{}. The blue and red disks represent clique positions along the variable gadget and the connector cliques, respectively.}
        \label{fig:marx-gadget}
    \end{figure}

    \subparagraph{Overview of the reduction for \MIF{}.}
    Our reduction preserves the two mechanisms of Marx's construction. Each variable $\ve{v}$ is represented by a \emph{variable gadget} $G_{\ve{v}}$: unit balls placed at clique positions arranged cyclically around the boundary of a box $Q_{\ve{v}}$, with per-value perturbations so that pairwise-disjoint selections around the cycle collapse to a single value of $C_{\ve{v}}$ (\emph{independent selection}, \Cref{clm:mif-cyclic}). Adjacent gadgets are tied by two connector chains, an entry chain and an exit chain, whose opposite perturbations enforce a $\ge$- and a $\le$-relation on the shared coordinate and hence, together, equality (\emph{coordinate consistency}, \Cref{clm:mif-connector}).

    In the independent-set setting these two properties come for free from cliques and disjointness. In the induced-forest setting they do not: an induced forest may contain two clique balls at the same position (as long as no triangle appears), and two intersecting balls at consecutive positions carry no penalty on their own. We restore both properties by adding two auxiliary balls at every clique position $p_j$: a \emph{private ball} $h_j(\ve{v})$, placed just outside $Q_{\ve{v}}$, and a \emph{forest ball} $t_j(\ve{v})$, placed just inside; each intersects exactly the clique balls at $p_j$. The private ball enforces the first property: together with two clique balls at $p_j$ it would close a triangle, so at most one clique ball per position can be selected alongside it (\Cref{clm:mif-private}). The forest balls enforce the second: we link all of them into a single \emph{global forest} $\mathcal{F}$, so that two selected intersecting balls at consecutive positions would create a cycle through the path joining their positions in $\mathcal{F}$ (\Cref{clm:mif-forest-forces}). Selections at consecutive positions are thus forced to be disjoint, which reinstates \Cref{clm:mif-cyclic} within a gadget and \Cref{clm:mif-connector} across connectors, exactly as in Marx's reduction. The global forest is assembled by linking the per-gadget trees along a fixed spanning tree $T$ of the primal graph. See \Cref{fig:mif.intersection.graph}. We give the construction first for $d = 2$ and then extend it to $d \ge 3$ using a Hamiltonian cycle on the cross-polytope of the box boundary.

    \subparagraph{Construction.}
    Let $d \ge 2$ be a fixed integer. As mentioned before, we use reduction from \textsc{Tiling}~CSP (\Cref{thm:pgcsp-to-tcsp}) to prove this theorem. Let $I=(V,\domain,\constraints)$ be an instance of $d$-dimensional \tcsp{} with $D = [\Lambda]^d$ for a positive integer $\Lambda$. Let $\loc : V \rightarrow V(G_k^d)$ be a subgraph isomorphism  between the primal graph $G=(V,E)$ of $I$ and an induced subgraph of the grid graph $G_k^d$ for a positive integer $k$. The set $\constraints$ contains the unary constraints $C_{\ve v} \subseteq [\Lambda]^d$, and for each pair of adjacent variables $\ve{a},\ve{b} \in V$ with $\loc(b) = \loc(a) + \dir{e}_i$ ($i \in [d]$), the implicit binary equality constraint $\ve{x}(i) = \ve{y}(i)$ is defined with $\ve{x}\in C_\ve{a}$ and $\ve{y} \in C_\ve{b}$. Let $\totalconstraints$ denote the total number of constraints (unary and binary). We construct a set $\mathcal{B}$ of unit balls in~$\R^d$ such that the~size of the \MIF{} in the intersection graph of~$\mathcal{B}$ encodes the optimal value of~$I$. We denote by $\textsc{mif}(\mathcal B)$ the maximum induced-forest size of the intersection graph of~$\mathcal B$. Let $\rho = 1/(400\,\Lambda^{d+2})$ and let $\iota : [\Lambda]^d \to \{0,1,\dots,\Lambda^d-1\}$ be an arbitrary bijection. Fix a spanning tree $T$ of $G$.\footnote{We assume that the primal graph $G$ is connected. Otherwise, the same construction can be applied to each connected component separately.} 

    \begin{figure}
        \centering
        \includegraphics[width=0.75\linewidth,page=9]{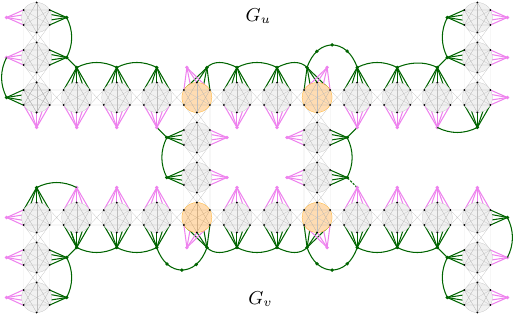}
        \caption{The vertices corresponding to clique balls (gray regions), forest balls (depicted in green), private balls (depicted in pink) are represented in the intersection graph of the instance in the neighborhood of two vertically adjacent variable gadgets $G_\ve{u}$ and $G_\ve{v}$. The dashed edge in this figure represents a potential edge of the spanning tree $T$ used to obtain the global forest.}
        \label{fig:mif.intersection.graph}
    \end{figure}

    \medskip\textbf{Case $d = 2$}: For each variable $\ve{v} \in V$, we construct a \emph{variable gadget} $G_{\ve{v}}$. The gadget is organized around an axis-aligned square $Q_{\ve{v}}$ of side length~$11$ centered at $c_{\ve{v}} := (14-2\rho)\loc(\ve{v}) \in \R^2$. On each side of $Q_{\ve{v}}$ we place $12$ \emph{clique positions} at unit spacing along the boundary, with the four corner positions shared between adjacent sides, yielding $\tau = 4 \times 12 - 4 = 44$ distinct positions. Let $p_1, p_2, \dots, p_\tau$ denote these positions listed consecutively around the boundary. Among these positions, $8$ are designated as \emph{port positions} (two per side): the fifth and eighth positions on each side. We call the positions $p_j$ for $j \in X^\downarrow:=\{5,16,27,38\}$ \emph{entry} ports and positions $p_j$ for $j \in X^\uparrow:=\{8, 19, 30, 41\}$ \emph{exit} ports (the orange points in \Cref{fig:mif-gadget}). Unlike the non-port positions, which lie exactly on the boundary of $Q_{\ve{v}}$, each port position is shifted slightly inward into $Q_\ve{v}$. On each side of $Q_{\ve{v}}$ with inward unit normal $\vec{n}$, the port positions are shifted by $\rho\vec{n}$.

    \begin{figure}
        \centering
        \includegraphics[width=0.5\linewidth,page=5]{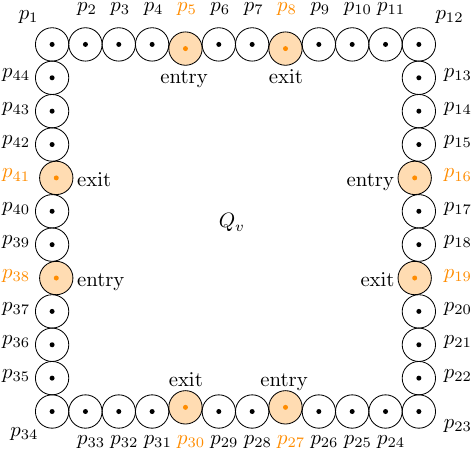}
        \caption{The clique positions along a square of side length $11$ ($12$ points per side with unit spacing). The orange points are the entry and exit port positions.}
        \label{fig:mif-gadget}
    \end{figure}
    
    Note that by construction, on the facing side of the adjacent gadget $G_{\ve{u}}$ with $\loc(\ve{u}) = \loc(\ve{v})+\dir{e}_i$, the exit port is aligned with the entry port of $G_{\ve{v}}$, and the entry port is aligned with the exit port of $G_{\ve{v}}$. The distance between the corresponding port positions from adjacent gadgets is exactly $3$. See \Cref{fig:mif-connector}(a).

    Our collection of unit balls $\cB$ constructed from $I$ contains five structures: clique balls, connector balls, private balls, forest balls, and the global forest. Next, we explain each separately.

    \smallskip
    \noindent\emph{Clique balls.}
    For each corner clique position $p_j$ for $j \in \{1, 12, 23, 34\}$, and each $\ve{s} \in C_{\ve{v}}$, we place a unit disk $B_j(\ve{v}, \ve{s})$ centered at $p_j + \iota(\ve{s})\rho\, (\vec{d}_j + \vec{d}'_j)$ where $\vec{d}_j = p_{j+1} - p_j$ and $\vec{d}'_j=p_j - p_{j-1}$. For the remaining clique positions $p_j$, define $\vec{d}_j$ as the unit vector along the boundary~of~$Q_v$ at $p_j$ in the direction of increasing index and let $\vec{n}$ be the inward unit normal of~$Q_\ve{v}$ on the side where $p_j$ is located. At each non-port position $p_j$, $j \in [\tau] \setminus (X^\downarrow \cup X^\uparrow)$, for each $\ve{s} \in C_{\ve{v}}$, we place a unit disk $B_j(\ve{v}, \ve{s})$ centered at $p_j + \iota(\ve{s})\rho\, \vec{d}_j$. At each port clique position $p_j$, for each $\ve{s} = (x,y) \in C_{\ve{v}}$, we place a unit disk $B_j(\ve{v}, \ve{s})$ centered~at
    \begin{itemize}
        \item $p_j + \iota(\ve{s})\rho\, \vec{d}_j + y\rho\,\vec{n}$ if $j \in \{5,27\}$,
        \item $p_j + \iota(\ve{s})\rho\, \vec{d}_j + x\rho\,\vec{n}$ if $j \in \{16,38\}$,
        \item $p_j + \iota(\ve{s})\rho\, \vec{d}_j - y\rho\,\vec{n}$ if $j \in \{8,30\}$, or
        \item $p_j + \iota(\ve{s})\rho\, \vec{d}_j - x\rho\,\vec{n}$ if $j \in \{19,41\}$.
    \end{itemize}
    
    In general, on the side facing direction $\dir{e}_i$, the entry port uses a perturbation of the form $+\ve{s}(i)\rho\,\vec{n}$ and the exit port uses $-\ve{s}(i)\rho\,\vec{n}$. The disks at each clique position form a clique in the intersection graph of~$\mathcal{B}$. The disks from cliques at different positions are disjoint unless they are at consecutive positions. We show that the intersection pattern of unit disks in this construction has the \emph{cyclic ordering} property described in the following claim.

    \begin{clm}\label{clm:mif-cyclic}
        For any two consecutive clique positions $p_j$, $p_{j+1}$ along the cycle (where one may be a port position shifted inward by $\rho$) and any $\ve{s}, \ve{s}' \in C_{\ve{v}}$ with $r = \iota(\ve{s})$, $r' = \iota(\ve{s}')$, the unit disks $B_j(\ve{v},\ve{s})$ and $B_{j+1}(\ve{v},\ve{s}')$ intersect if and only if $r > r'$.
    \end{clm}
    \begin{proof}
        Let $\vec{d}$ be the common unit vector at $p_j$ and $p_{j+1}$ along the boundary of~$Q_v$ (both $p_j$ and $p_{j+1}$ are on the same side of $Q_\ve{v}$), and let $\vec{n}$ be the inward normal, so $\vec{d}$ and $\vec{n}$ are perpendicular. Before the inward shift of the port positions, $\vec{d} = p_{j+1} - p_j$ (unit spacing). The contribution of the port-specific perturbation in the direction of the normal $\vec{n}$ to the distance between consecutive balls on the same side is at most $3\Lambda^d\rho$ ($\pm \ve{s}(i)\rho\,\vec{n} + \rho\,\vec{n} + \Lambda^d\rho\,\vec{n}$), so the squared distance satisfies
        \[
            \Delta^2 \le \bigl(1 + (r'-r)\rho\bigr)^2 + 9\Lambda^{2d}\rho^2.
        \]
        If $r \le r'$, then $1 + (r'-r)\rho \ge 1$, giving $\Delta^2 \ge 1$ and the open disks are disjoint.
        If $r > r'$, then $r - r' \ge 1$ and the remaining terms are at most $9\Lambda^{2d}\rho^2$, so
        \[
            \Delta^2 \le (1-\rho)^2 + 9\Lambda^{2d}\rho^2 \le 1 - 2\rho + 10\Lambda^{2d}\rho^2.
        \]
        Since $\rho = 1/(400\,\Lambda^{d+2})$, we have $10\Lambda^{2d}\rho^2 \ll \rho$, hence $\Delta^2 < 1$. The disks intersect.
    \end{proof}

    By the cyclic ordering property (\Cref{clm:mif-cyclic}), if balls $B_j(\ve{v},\ve{s})$ and $B_{j+1}(\ve{v},\ve{s}')$ from consecutive clique positions are both selected and are disjoint, then $\iota(\ve{s}) \le \iota(\ve{s}')$. Applying this around the entire cycle of gadget $Q_\ve{v}$ yields $\iota(\ve{s}) \le \iota(\ve{s}') \le \cdots \le \iota(\ve{s})$, so if the selected balls on clique positions are pairwise disjoint, all selected balls correspond to the same value of~$C_\ve{v}$.

    \begin{figure}
        \centering
        \includegraphics[width=0.9\linewidth,page=4]{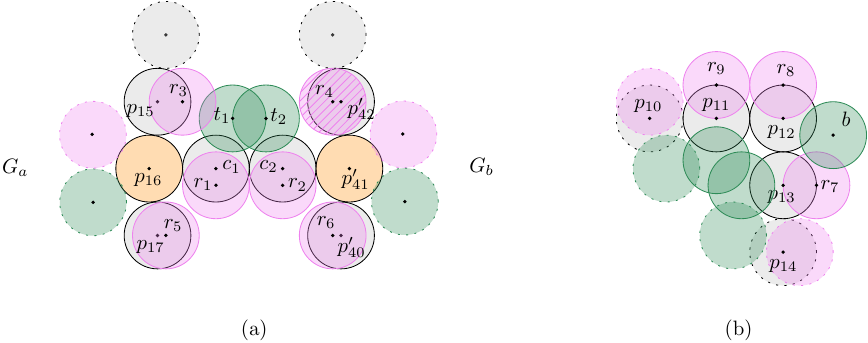}
        \caption{(a) Illustration of the entry chain from entry port position at $p_{16}$ in $G_\ve{a}$ to the exit port position at $p'_{41}$ in $G_\ve{b}$ for adjacent variables $\ve{a}$ and $\ve{b}$. The distance of entry port $p_{16}$ in $G_\ve{a}$ from the exit port $p'_{41}$ in $G_\ve{b}$ is exactly $3$. Two intermediate connector positions are centered at $c_1$ and $c_2$ with unit spacing from the port positions. For entry chain the centers of the private balls are placed below the clique positions (centers $r_1$ and $r_2$). The two forest balls centered at $t_1$ and $t_2$ intersect each other and the clique balls from their corresponding clique position. The private ball of $p_{15}$ centered at $r_3$ is placed slight outward to intersect the forest balls of the first intermediate connector position. If the edge $\{\ve{a},\ve{b}\}$ appears in the spanning tree $T$ of $G$ then the private ball of the clique position $p'_{42}$ centered at $r_4$ will be positioned slightly outward so that it serves as a connector of $T_\ve{a}$ and $T_\ve{b}$ in the global forest. (b) Illustration of placement of a forest ball on the corner of a variable gadget. The forest ball of corner position~$p_{12}$ intersects the private ball of the clique position~$p_{13}$ centered at~$r_7$.}
        \label{fig:mif-connector}
    \end{figure}

    \begin{figure}
        \centering
        \includegraphics[width=0.85\linewidth,page=6]{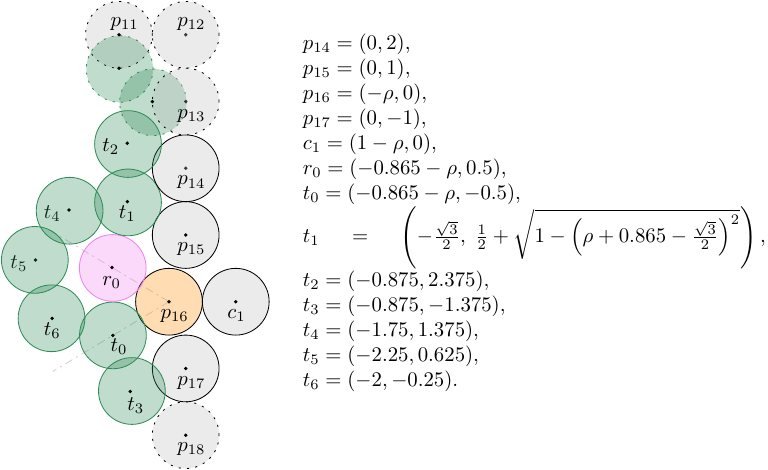}
        \caption{Demonstration of relative coordinates of unit disks at the neighborhood of entry port at position $p_{16}$. The port position $p_{16}$ is placed slightly inward such that the private ball centered at $r_0$ and the forest ball centered at $t_0$ can both be placed inside the square $Q_\ve{v}$. Note that the coordinates are chosen in a way that the unit disk centered at $r_0$ only intersects the clique balls from the position $p_{16}$. Moreover the balls centered at $t_1,t_2,t_4,t_5,t_6,t_0,t_3$ form a subpath of the tree $T_\ve{v}$. The coordinates of $t_1$ are set such that it intersects all the clique balls in $p_{15}$ and it is tangent to the private ball centered at $r_0$ (but not intersecting it as the disks are open).}
        \label{fig:mif3}
    \end{figure}

    \smallskip
    \noindent\emph{Connectors.}
    For each edge $\{\ve{a}, \ve{b}\} \in E$ with $\loc(\ve{b}) = \loc(\ve{a}) + \dir{e}_i$ for some $i \in [d]$, we connect gadgets $G_{\ve{a}}$ and $G_{\ve{b}}$ through the two port pairs on the relevant sides. Let $p^{\ve{a}}_1, p^{\ve{a}}_2$ be the entry and exit port positions on the side of $G_{\ve{a}}$ facing $G_{\ve{b}}$, and $p^{\ve{b}}_1, p^{\ve{b}}_2$ be the corresponding exit and entry positions on the facing side of $G_{\ve{b}}$, respectively.
 
    We build a chain of $2$ intermediate cliques between each aligned pair: one \emph{entry chain} from $p^{\ve{a}}_1$ (entry of $G_{\ve{a}}$) to $p^{\ve{b}}_1$ (exit of $G_{\ve{b}}$), and one \emph{exit chain} from $p^{\ve{a}}_2$ (exit of $G_{\ve{a}}$) to $p^{\ve{b}}_2$ (entry of $G_{\ve{b}}$). Note that the entry chain of $G_\ve{a}$ is an exit chain of $G_\ve{b}$ and vice versa. Each intermediate clique has $\Lambda$~balls, one for each value $q \in [\Lambda]$ of the $i$-th coordinate.
    \begin{itemize}
        \item In the \textbf{entry chain}, for each $q \in [\Lambda]$, we place two balls centered at
        \[
            p^{\ve{a}}_1 + \dir{e}_i - q\rho\,\dir{e}_i
            \quad\text{and}\quad
            p^{\ve{a}}_1 + 2\dir{e}_i - q\rho\,\dir{e}_i.
        \]
        \item In the \textbf{exit chain}, for each $q \in [\Lambda]$, we place two balls centered at
        \[
            p^{\ve{a}}_2 + \dir{e}_i + q\rho\,\dir{e}_i
            \quad\text{and}\quad
            p^{\ve{a}}_2 + 2\dir{e}_i + q\rho\,\dir{e}_i.
        \]
    \end{itemize}
    The critical asymmetry is that the entry chain has perturbation $-q\rho$ (opposing the chain direction $\dir{e}_i$), while the exit chain has perturbation $+q\rho$ (along the chain direction).
 
    \begin{clm}\label{clm:mif-connector}
        In the entry chain between $G_{\ve{a}}$ and $G_{\ve{b}}$ ($\loc(\ve{b}) = \loc(\ve{a})+\dir{e}_i$), consecutive balls corresponding to the values $r_\ell \in [\Lambda]$ and $r_{\ell+1} \in [\Lambda]$ (ordered from $G_{\ve{a}}$ toward $G_{\ve{b}}$) are disjoint if and only if $r_\ell \ge r_{\ell+1}$.
        In the exit chain, consecutive balls corresponding to the values $r_\ell \in [\Lambda]$ and $r_{\ell+1} \in [\Lambda]$ (ordered from $G_{\ve{a}}$ toward $G_{\ve{b}}$) are disjoint if and only if $r_\ell \le r_{\ell+1}$.
    \end{clm}
    \begin{proof}
        Consider the entry chain. All balls (port and intermediate) have perturbation in the $-\dir{e}_i$ direction: the entry port of $G_{\ve{a}}$ has perturbation $-\ve{s}(i)\rho$ in the $\dir{e}_i$ direction for each $\ve{s} \in C_\ve{a}$ (from the $+\ve{s}(i)\rho\,\ve{n}$ term, since the inward normal $\ve{n}$ equals $-\dir{e}_i$ for the side facing $+\dir{e}_i$), and each intermediate ball has perturbation $-q\rho$ in the $\dir{e}_i$ direction. Similarly, the exit port of $G_{\ve{b}}$ (on the side with inward normal $+\dir{e}_i$) has perturbation $-\ve{t}(i)\rho$ in the $\dir{e}_i$ direction for each $\ve{t} \in C_\ve{b}$.

        Consecutive positions are separated by $1$ in the $\dir{e}_i$ direction. Writing $r_\ell$ for the value that encodes the perturbation at position $\ell$ and $r_{\ell+1}$ at position $\ell+1$, the distance in the $\dir{e}_i$ direction is $1 - (r_{\ell+1} - r_\ell)\rho$, and the perpendicular components contribute at most $O(\Lambda^{d} \rho)$ (from the tangent perturbation $\iota(\ve{s})\rho$ at port positions). Thus
        \[
            \Delta^2 = \bigl(1 - (r_{\ell+1}-r_\ell)\rho\bigr)^2 + c_0\Lambda^{2d}\rho^2.
        \]
        for a constant $c_0$. If $r_{\ell+1} \le r_\ell$, then $1 - (r_{\ell+1}-r_\ell)\rho \ge 1$, so $\Delta \ge 1$ and the balls are disjoint.
        If $r_{\ell+1} > r_\ell$, then $r_{\ell+1}-r_\ell \ge 1$, so $\Delta^2 \le (1-\rho)^2 + c_0\Lambda^{2d}\rho^2 < 1$ and the balls intersect ($\Lambda^{2d}\rho^2 \ll \rho)$. Similar argument proves the claim for the exit chain.
    \end{proof}

    \smallskip
    \noindent\emph{Private balls.}
    For each clique position $p_j$ in the gadget $G_v$, we place one unit disk $h_j(\ve{v})$, called the \emph{private ball}, with its center slightly exterior to the boundary square~$Q_v$. The port positions along the boundary of $Q_\ve{v}$ are the only case in which we need to place the private balls inside the square $Q_\ve{v}$ (the inward shift of each port position creates room for this). Similarly, for each clique at the connectors, we dedicate a private ball specific to that clique (see \Cref{fig:mif-connector}(a)). The private ball $h_j(\ve{v})$ intersects every clique ball at position $p_j$ and it does not intersect any other private ball or any clique ball at a different position (See \Cref{fig:mif.gadget}). A private ball may intersect a ``forest ball'' from an adjacent clique position or a connector position (this will be explained in detail in the following parts).

    \smallskip  
    \noindent\emph{Forest balls.}
    For each clique position $p_j$, we place one additional unit disk $t_j(\ve{v})$, called the \emph{forest ball}, with center slightly interior to the boundary square. The forest ball $t_j(\ve{v})$ intersects every clique ball at position $p_j$ and no clique ball at any other position. Moreover, we build a tree $T_\ve{v}$ in the intersection graph of $G_\ve{v}$ such that $T_\ve{v}$ covers the $\tau$ forest balls within $G_\ve{v}$. This is achieved by positioning the forest balls within $G_\ve{v}$ in a way that their mutual intersection pattern forms spanning paths in $T_\ve{v}$ and by routing extra balls along paths in the interior of the square. Since the square has ample interior space this can be done with $O(1)$ additional forest balls serving as intermediate vertices in $T_{\ve{v}}$ (See \Cref{fig:mif.gadget}). These additional balls intersect only their forest neighbors and no clique or private balls.
    
    There is one exception in the placement of forest balls. For the corner clique positions due to space limitations, we place the forest ball outside $Q_\ve{v}$. Let $b$ be the forest ball of the corner clique positioned at $p_j$. In this particular case, the private ball of the neighboring clique positioned at $p_{j+1}$ intersects the forest ball $b$ in order to include $b$ in the tree $T_\ve{v}$ (See \Cref{fig:mif-connector}(b)). As explained earlier, the private balls can be part of the tree $T_\ve{v}$ in order to connect the tree components near clique positions. Let $g$ denote the total number of forest balls in gadget $G_{\ve{v}}$.
 
    Each of the $2$ intermediate clique positions in every connector also receives one forest ball and one private ball, following the same rules as in the gadgets: the forest ball intersects all clique balls at its position and its neighboring forest ball, and the private ball intersects only the clique balls at its position (See \Cref{fig:mif-connector}(a)).


    \smallskip
    \noindent\emph{Global forest.}
    We connect the trees of all gadgets and connectors into a single tree $\mathcal{F}$ as follows. For each edge $\{\ve{a}, \ve{b}\} \in E$, the forest balls of each connector chain (entry and exit) between $G_{\ve{a}}$ and $G_{\ve{b}}$ intersect so that they induce a path on two vertices (e.g. forest balls $t_1$ and $t_2$ in \Cref{fig:mif-connector}(a)). Next, we connect these connector paths into the neighboring trees $T_\ve{a}$ and $T_\ve{b}$ in the gadgets $G_\ve{a}$ and $G_\ve{b}$ as follows.
 
    \begin{itemize}
        \item \textbf{Spanning tree edges} ($\{\ve{a},\ve{b}\} \in E(T)$): We designate one chain (say the entry chain) as the \emph{bridge}. The forest ball path of this chain connects $T_{\ve{a}}$ to $T_{\ve{b}}$, via neighboring private balls. The other chain (exit) is attached as a \emph{pendant subtree} to $T_{\ve{a}}$ without connecting to $T_{\ve{b}}$.
        \item \textbf{Non-spanning tree edges} ($\{\ve{a},\ve{b}\} \notin E(T)$): Both connector chains are attached as pendant subtrees to $T_{\ve{a}}$ via a private ball on $T_{\ve{a}}$.
    \end{itemize}

    \begin{figure}
        \centering
        \includegraphics[width=0.95\linewidth,page=1]{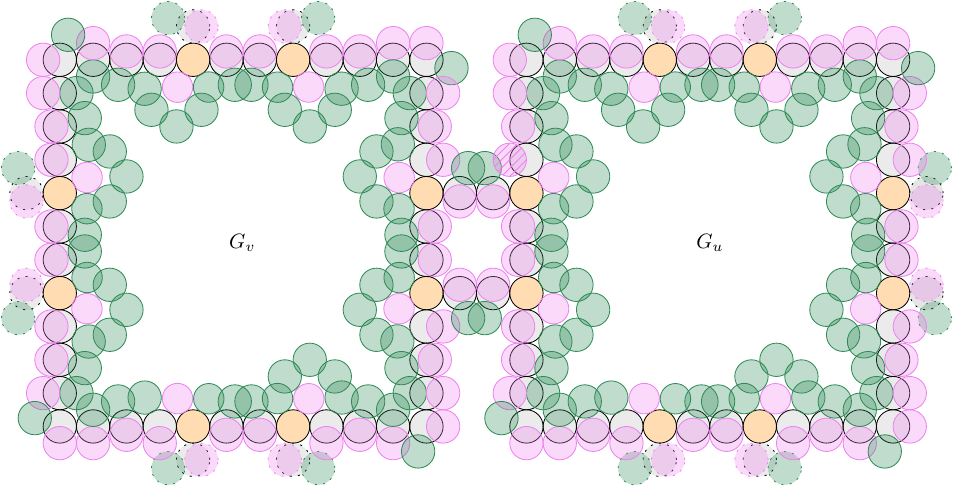}
        \caption{An instance of unit disks obtained from two adjacent variables of a $2$-dimensional \tcsp{} where their corresponding edge appears in the spanning tree $T$. The green, violet and orange colors represent the forest balls, private balls, port positions, respectively. The violet colored disk with parallel stripes represents the private ball in $G_u$ that ensures the connectivity of the global forest.}
        \label{fig:mif.gadget}
    \end{figure}
 
    Since $T$ spans $G$, using one bridge per edge in $T$ connects all gadget trees. The pendant subtrees do not create cycles in~$\mathcal{F}$. In order to merge the forest ball chain of a connector to $T_\ve{a}$, we use the private ball of the clique position in $G_\ve{a}$ adjacent to the entry/exit port of the chain along the sequence of positions. This can be achieved by placing the private ball in a way that intersects both the corresponding clique disks and the forest ball of the first intermediate clique along the connector chain. For example, in \Cref{fig:mif-connector}(a), the private ball centered at $r_3$ intersects both the clique at position $p_{15}$ and the forest ball centered at $t_1$. Moreover, if the edge $\{\ve{a},\ve{b}\} \in E(T)$, the disk centered at $r_4$ can be placed slightly outward the gadget $Q_\ve{b}$ such that it intersects the disk centered at $t_2$.

    \Cref{fig:mif3} gives relative coordinates of different types of balls in the neighborhood of an entry port position $p_{16}$. The coordinates of the remaining balls are either described earlier or can be obtained from those listed in \Cref{fig:mif3} by rotating and reflecting. \Cref{fig:mif.intersection.graph} shows the subgraph of the intersection graph of the instance constructed from $I$ at the neighborhood of two adjacent variables.

    \medskip\textbf{Case $d\ge 3$.}
    We now describe the construction in full for $d \ge 3$. The gadget for each variable $\ve{v} \in V$ is organized around an axis-aligned $d$-dimensional box $Q_{\ve{v}}$ of side length $30$ centered at $c_{\ve{v}} := 33\loc(\ve{v}) \in \IR^d$. The box $Q_{\ve{v}}$ has $2d$ facets. Let $\boldsymbol{\sigma} = \boldsymbolvec{\sigma}_1, \boldsymbolvec{\sigma}_2,\dots, \boldsymbolvec{\sigma}_{2d}$ be the outward normals of these facets where $\boldsymbolvec{\sigma}_i = \dir{e}_i$ for $1 \le i \le d$ and $\boldsymbolvec{\sigma}_i = -\dir{e}_{i - d}$ for $d+1 \le i \le 2d$. Note that the ordering of the normals in $\boldsymbolvec{\sigma}$ follows a Hamiltonian cycle of the cross-polytope whose vertices correspond to the facets of the box $Q_\ve{v}$. The facet with outward normal $\boldsymbolvec{\sigma}_i$ is denoted $F_i$. We place clique positions on the boundary of $Q_{\ve{v}}$ in a cyclic order determined by a Hamiltonian cycle of the \emph{cross-polytope} of $Q_\ve{v}$ similar to \cite{Kisfaludi-BakMZ19}, and define clique, private, forest, and connector balls exactly as in the $d=2$ case but with the geometry adapted to $d$ dimensions.
    
    \begin{figure}
        \centering
        \includegraphics[width=0.7\linewidth,page=2]{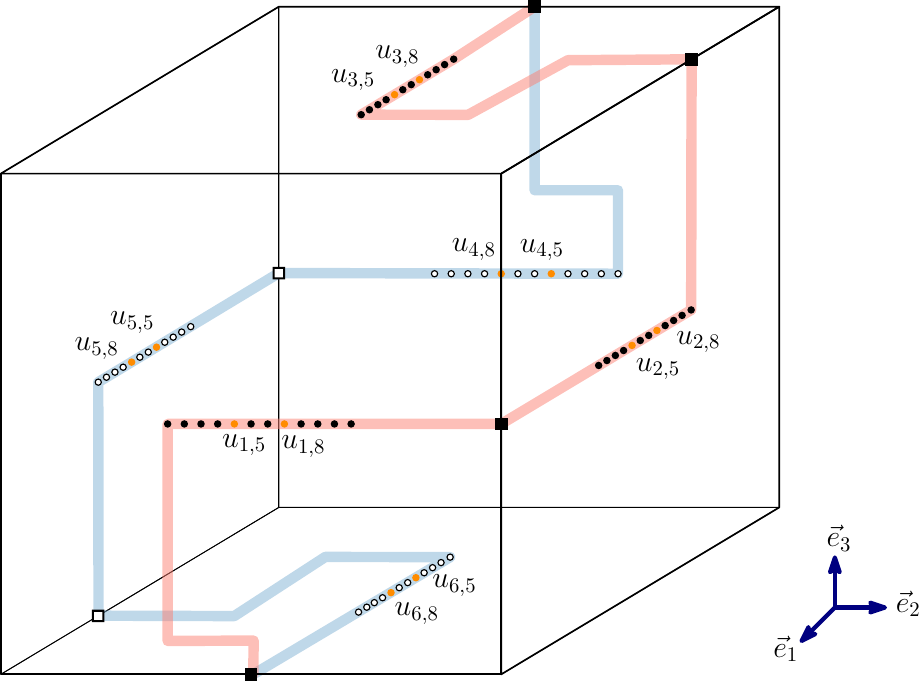}
        \caption{Demonstration of the Hamiltonian cycle $H_\ve{v}$ on the boundary of $Q_\ve{v}$ for~$d = 3$. The corner points $u'_i$ for $i \in [6]$ are represented by small squares. The solid points and squares represent points on the facets $F_1$, $F_2$, and $F_3$, and the empty points and squares represent the point on the facets~$F_4$, $F_5$, and $F_6$.}
        \label{fig:ham-cycle}
    \end{figure}

    \noindent\emph{Hamiltonian cycle on the boundary.}
    Recall that the ordering $\boldsymbolvec{\sigma}_1, \boldsymbolvec{\sigma}_2, \ldots, \boldsymbolvec{\sigma}_{2d}$ is a Hamiltonian cycle of the cross-polytope: consecutive normals $\boldsymbolvec{\sigma}_i$ and $\boldsymbolvec{\sigma}_{i+1}$ (indices modulo $2d$) are orthogonal.  We say that two points $\ve{x}, \ve{y} \in \IR^d$ are adjacent if there is an index $ i \in \{1,\dots,d\}$ such that $\ve y = \ve x + \dir{e}_i$ or $\ve x = \ve y + \dir{e}_i$. For each facet $F_i$, we place 12 consecutively ordered adjacent points $u_{i,1}, u_{i,2}, \dots, u_{i,12}$ on $F_i$. Again we call the point $u_{i,5}$ the \emph{entry port} and the point $u_{i,8}$ the \emph{exit port}. The coordinates are defined as follows with respect to the center of $Q_\ve{v}$:
    \begin{itemize}
        \item $u_{1,j} = c_v + 15\boldsymbolvec{\sigma}_1 + (j-6)\dir{e}_2$ for $j = 1,\ldots,12$,
        \item $u_{i,j} = c_v + 15\boldsymbolvec{\sigma}_i + (7-j)\dir{e}_1$ for $2 \le i \le d$ and $j = 1,\ldots,12$,
        \item $u_{d+1,j} = c_v + 15\boldsymbolvec{\sigma}_{d+1} + (7-j)\dir{e}_2$ for $j = 1,\ldots,12$,
        \item $u_{i,j} = c_v + 15\boldsymbolvec{\sigma}_i + (j-6)\dir{e}_1$ for $d+2 \le i \le 2d$ and $j = 1,\ldots,12$.
    \end{itemize}
    Observe that the entry and exit ports of facet $F_i$ are aligned with the exit and entry ports of the opposite facet $F_{i+d}$, for each $1 \le i \le d$.
 
    Next, for each $1 \le i \le 2d$, we connect $u_{i,12}$ to $u_{i+1,1}$ (indices modulo $2d$) by a rectilinear directed subpath: starting from $u_{i,12}$, we add constant many consecutively adjacent points in direction $\boldsymbolvec{\sigma}_{i+1}$ until reaching a point $u'_{i+1}$ on the facet $F_{i+1}$. Within each facet $F_i$, we connect $u'_i$ to $u_{i,1}$ using adjacent points via a rectilinear directed subpath $\pi_i$ of minimal length. We choose the path $\pi_i$ so that every point of $\pi_i$ is at $\ell_1$-distance at least $2$ from all other previously placed points on $F_i$ (except for $u_{i,1}$ and $u'_{i}$). Let $H_{\ve{v}}$ denote the resulting directed Hamiltonian cycle on the boundary of $Q_{\ve{v}}$, and let $p_1, p_2, \ldots, p_\tau$ be its points in cyclic order (See \Cref{fig:ham-cycle}). Since each facet contributes $O(1)$ points, we have $\tau = O(d)$.
    
    A directed subpath of $H_{\ve{v}}$ is called \emph{maximal} if all its points lie on the same axis-aligned line (differ in exactly one coordinate) and cannot be extended. Each maximal directed subpath on facet $F_i$ has a directional unit vector $\vec{a} = p_{j+1} - p_j$ (for two consecutive points $p_j$ and $p_{j+1}$ on the subpath) and $\boldsymbolvec{\sigma}_i$ as the outward facet normal. A position $p_j$ on $H_\ve{v}$ is called a corner position if the two vectors $\vec{d}_j = p_{j+1}-p_j$ and $\vec{d}'_j=p_j - p_{j-1}$ are perpendicular.
     
    Note that unlike the case $d=2$, for $d\ge3$ we do not need to shift the port positions inwards as there is enough space to position the private and forest balls.

    \noindent\emph{Clique balls for $d \ge 3$.}
    The clique positions are the $\tau$ points of $H_{\ve{v}}$. Among these, the points $u_{i,5}$ and $u_{i,8}$ on each facet $F_i$ are designated as the \emph{entry port} and \emph{exit port} of facet $F_i$, respectively.
 
    At each corner clique position, similar to the $d = 2$ case, for each $\ve{s} \in C_{\ve{v}}$, we place a unit ball $B_j(\ve{v}, \ve{s})$ centered at $p_j + \iota(\ve{s})\rho\, (\vec{d}_j + \vec{d}'_j)$ where $\vec{d}_j = p_{j+1} -p_j$ and $\vec{d}'_j=p_j - p_{j-1}$. At each non-port clique position $p_j$ other than corners, for each $\ve{s} \in C_{\ve{v}}$, we place a unit ball $B_j(\ve{v}, \ve{s})$ centered at $p_j + \iota(\ve{s})\rho\,\vec{a}_j$, where $\vec{a}_j$ is the directional unit vector of the maximal directed subpath containing~$p_j$.
    
    At each port position on facet $F_i$ ($1 \le i \le 2d$), let $\ell \in [d]$ be the coordinate such that $\boldsymbolvec{\sigma}_i = \pm\dir{e}_\ell$ (i.e., $\ell = i$ if $i \le d$ and $\ell = i - d$ if $i > d$). For each $\ve{s} \in C_{\ve{v}}$, the clique ball $B_j(\ve{v},\ve{s})$ is centered at
    \[
        p_j + \iota(\ve{s})\rho\,\vec{a}_j - \ve{s}[\ell]\,\rho\,\boldsymbolvec{\sigma}_i \qquad\text{(entry port)}, \qquad
        p_j + \iota(\ve{s})\rho\,\vec{a}_j + \ve{s}[\ell]\,\rho\,\boldsymbolvec{\sigma}_i \qquad\text{(exit port)}.
    \]
    That is, the entry port on the side facing direction $\boldsymbolvec{\sigma}_i$ uses perturbation $-\ve{s}[\ell]\rho\,\boldsymbolvec{\sigma}_i$ and the exit port uses $+\ve{s}[\ell]\rho\,\boldsymbolvec{\sigma}_i$. Note that the cyclic ordering property (\Cref{clm:mif-cyclic}) holds verbatim for $d\ge3$.

    \noindent\emph{Connectors for $d \ge 3$.}
    For each edge $\{\ve{a},\ve{b}\} \in E$ with $\loc(\ve{b}) = \loc(\ve{a})+\dir{e}_i$, the connector structure is the same as for $d = 2$. The gadget $G_{\ve{a}}$ has entry and exit ports on facet $F_i$ (with normal $\boldsymbolvec{\sigma}_i = \dir{e}_i$), aligned with the exit and entry ports on facet $F_{i+d}$ (with normal $\boldsymbolvec{\sigma}_{i+d} = -\dir{e}_i$) of $G_{\ve{b}}$. We build an entry chain and an exit chain between each aligned pair, with $2$ intermediate cliques of $\Lambda$ balls each. The entry chain uses perturbation $-q\rho\,\dir{e}_i$ (opposing the direction from $G_{\ve{a}}$ to $G_{\ve{b}}$), enforcing $\ve{s}(i) \ge \ve{t}(i)$ for $\ve{s} \in C_\ve{a}$ and $\ve{t} \in C_\ve{b}$. The exit chain uses perturbation $+q\rho\,\dir{e}_i$, enforcing $\ve{s}(i) \le \ve{t}(i)$. Observe that \Cref{clm:mif-connector} holds for the $d \ge 3$ case with the same proof.

    \noindent\emph{Private balls, forest balls, and global forest for $d \ge 3$.}
    At each clique position $p_j$ (in a gadget or connector), we place one private ball $h_j(\ve{v})$ and one forest ball $t_j(\ve{v})$, following the same rules as in $d = 2$:
    \begin{itemize}
        \item The private ball is placed slightly exterior to $Q_{\ve{v}}$ (i.e., placed from $p_j$ in the direction $+\boldsymbolvec{\sigma}_i$ for the facet $F_i$ containing $p_j$, or inward for port positions). It intersects all clique balls at $p_j$ and nothing else.
        \item The forest ball is placed slightly interior to $Q_{\ve{v}}$ (placed in the direction $-\boldsymbolvec{\sigma}_i$, toward the center). It intersects all clique balls at $p_j$ and no clique ball at any other position.
    \end{itemize}
    As in $d = 2$, the private ball and forest ball at the same position do not intersect each other (they are on opposite sides of the clique). Unlike the $d = 2$ case, at a corner position along the Hamiltonian cycle, there is enough room in the $d$-dimensional geometry to place private and forest balls without unwanted intersections.
    
    The forest balls within each gadget are connected into a tree $T_{\ve{v}}$ using $O(d)$ additional balls routed through the interior of $Q_{\ve{v}}$.
    Recall that we denote by $g = O(d)$ the total number of forest balls per variable gadget.
    
    The global forest $\mathcal{F}$ is constructed exactly as in $d = 2$: for tree edges of $T$, one connector chain serves as a bridge; for non-tree edges, both chains are pendant subtrees. This finishes the construction. We capture the structure of the induced forest by this construction in the next claims.

    \begin{clm}\label{clm:mif-private}
        Let $F$ be an induced forest of $\mathcal{B}$ that contains a private ball $h_j(\ve{v})$ or a forest ball $t_j(\ve{v})$ at position $p_j$. Then $F$ can contain at most one clique ball at position $p_j$.
    \end{clm}
    
    \begin{proof}
        Suppose $F$ contains two clique balls $B_j(\ve{v},\ve{s})$ and $B_j(\ve{v},\ve{s}')$ at the same position $p_j$. These two balls intersect. Both also intersect $h_j(\ve{v})$, which is in $F$ by assumption. This creates a triangle in the intersection graph of $F$, contradicting the assumption that $F$ is a forest. Replacing the role of $h_j(\ve{v})$ with $t_j(\ve{v})$ implies the same result.
    \end{proof}

    \begin{clm}\label{clm:mif-forest-forces}
        Let $F$ be an induced forest of $\mathcal{B}$ containing all forest balls of $\mathcal{F}$. Then clique balls in distinct clique positions are disjoint. Moreover, if $F$ is maximal, all private balls are in~$F$.
    \end{clm}
    \begin{proof}
        By construction, clique balls from non-adjacent clique positions do not intersect. For two adjacent clique positions $p_j$ and $p_{j+1}$ along $Q_\ve{v}$, if the forest $F$ contains a clique ball from each of the positions then they do not intersect. If they intersect then these two clique balls, together with the path between $t_j(\ve{v})$ and $t_{j+1}(\ve{v})$ create a cycle. A similar argument works for two adjacent connector positions. Note that the perturbations of the private balls are chosen so that they help realize the spanning tree~$T$ of the primal graph $G$ within the global forest $\mathcal{F}$. Therefore, a clique ball in a gadget adjacent to a connector ball in an entry (or exit) chain at that position does not intersect, as such an intersection would create a cycle in~$\mathcal{F}$. A private ball, other than a clique ball in its corresponding clique position, may intersect a forest ball either at a corner or along the connector chain. In the former case, the forest ball together with the private ball forms a pendant subtree of the global forest. In the latter case, the spanning tree $T$ of the primal graph $G$ guarantees that no new cycle is created. Therefore, any maximal induced forest that contains all forest balls must also include all private balls.
    \end{proof}

    \subparagraph{Exact Equivalence.}
    Assume $I$ has a solution $f : V \to [\Lambda]^d$ with $f(\ve{v}) \in C_{\ve{v}}$ for all $\ve{v}$.
    We construct an induced forest $F$ as follows.
    \begin{itemize}
        \item \textit{Forest balls:} include every forest ball.
        \item \textit{Private balls:} include every private ball.
        \item \textit{Clique representatives:} for each gadget $G_{\ve{v}}$ and each clique position $p_j$, include $B_j(\ve{v}, f(\ve{v}))$. For each intermediate clique in the connector between $\ve{a}$ and $\ve{b}$ with $\loc(\ve{b}) = \loc(\ve{a})+\dir{e}_i$, include the ball corresponding to the $i$-th coordinate of $f(\ve{a})$ and $f(\ve{b})$.
    \end{itemize}

    We verify that $F$ is a forest. The forest balls induce the forest $\mathcal{F}$. Since $f$ is satisfying, the cyclic ordering is respected everywhere: within each gadget, consecutive representatives correspond to the same domain value $\iota(f(\ve{v}))$ and hence do not intersect. Across each connector, the $i$-th coordinate values agree and hence consecutive representatives do not intersect by \Cref{clm:mif-connector}. Therefore, each representative $B_j(\ve{v}, f(\ve{v}))$ is adjacent in $F$ only to the forest ball $g_j(\ve{v})$, making it a pendant of~$\mathcal{F}$. Each private ball $h_j(\ve{v})$ is adjacent in $F$ only to $B_j(\ve{v}, f(\ve{v}))$, making it a pendant of the representative. Hence $F$ is a tree with pendants, which is a forest. The size of this forest is $|F| = (g + 2\tau)|V| + 12|E|$: there are $g$ forest balls, $\tau$ private balls, and $\tau$ clique balls per variable, and $4$ forest balls, $4$ private balls, and $4$ clique balls per binary constraint.

    Conversely, suppose $F^\star$ is an induced forest of size $(g + 2\tau)|V| + 12|E|$. We first establish an upper bound. Each private ball can be in $F^\star$ only if at most one clique ball at its position is also in $F^\star$ (by \Cref{clm:mif-private}, otherwise a triangle forms). Since each clique position contributes at most two balls to any maximum induced forest, the clique representatives together with private balls contribute at most $2\tau|V| + 8|E|$ balls. Each clique contributes either two clique balls or one clique ball and one private ball. The forest $\mathcal{F}$ has $g|V|+4|E|$ balls. Hence $|F^\star| \le (g + 2\tau)|V| + 12|E|$, and equality forces:
    \begin{itemize}
        \item (a) every forest ball is in $F^\star$,
        \item (b) every private ball is in $F^\star$, and
        \item (c) exactly one clique representative per clique position is in $F^\star$.
    \end{itemize}
    By (a), (c), and \Cref{clm:mif-forest-forces}, each representative is a pendant of $\mathcal{F}$ via its forest ball. Since $F^\star$ is a forest and $\mathcal{F}$ is a tree, if two consecutive representatives (along the gadget cycle or along a connector chain) were to intersect, they would create a cycle through the forest ball path connecting their positions in $\mathcal{F}$, contradicting the forest property. Therefore, consecutive representatives are disjoint.
    
    Within each gadget $G_{\ve{v}}$, the disjointness of consecutive representatives around the cycle forces all representatives to encode the same domain value $\ve{s}_{\ve{v}} \in C_{\ve{v}}$. Define $\bar f(\ve{v}) := \ve{s}_{\ve{v}}$.
    
    For each edge $\{\ve{a},\ve{b}\} \in E$ with $\loc(\ve{b}) = \loc(\ve{a})+\dir{e}_i$, consecutive representatives along the entry chain are disjoint. By the assignment $\bar f$, let $\bar f(\ve{a}) = (x_1,x_2,\dots,x_d)$ and $\bar f(\ve{b})=(y_1,y_2,y_3,\dots,y_d)$. By \Cref{clm:mif-connector} the $i$-th coordinate values satisfy $x_i \ge y_i$. Similarly, disjointness along the exit chain gives $x_i \le y_i$. Together, $x_i = y_i$, so every binary equality constraint is satisfied. The unary constraints are satisfied by construction since $\ve{s}_{\ve{v}} \in C_{\ve{v}}$. Hence $\bar f$ is a solution for $I$.

    Therefore, we conclude that $I$ has a solution $f$ if and only if the constructed instance $\cB$ has an induced forest of size at $(g + 2\tau)|V| + 12|E|$.
    
    \subparagraph{Error Analysis.} Let $\Gamma = (g + 2\tau)|V| + 12|E|$. Assume that every assignment to $I$ violates at least a $\delta$-fraction of the constraints. We show that every induced forest of the intersection graph of $\mathcal{B}$ has size at most $(1 - c'_d\,\delta)\,\Gamma$ for a constant $c'_d > 0$ depending only on $d$. Let $F^\star$ be a maximum induced forest of the intersection graph of $\mathcal{B}$. A variable gadget $G_{\ve{v}}$ is called \emph{full} if $F^\star$ contains all forest and private balls of $G_{\ve{v}}$ and all incident connectors, together with exactly one clique representative per clique position (within the gadget and its connectors), and all these representatives are mutually consistent: within the gadget, they correspond to a single value $\ve{s} \in C_{\ve{v}}$, and in each connector in direction $i$, they correspond to the value $\ve{s}(i)$. Otherwise we call that gadget \emph{partial}. Accordingly, we call a variable $v$ full (partial) if the gadget $G_v$ is full (resp. partial). The number of balls in each full gadget is stated in the following observation via simple counting.

    \begin{obs} \label{obs:full-gadget-size}
        Each full variable $\ve{v} \in V$ contributes $g + 2\tau + 12|E_\ve{v}|$ balls to $F^\star$ where $E_\ve{v}$ is the set of edges incident to $\ve{v}$ and each partial variable $\ve{u}$ contributes strictly less than $g + 2\tau + 12|E_\ve{u}|$.
    \end{obs}

    The following observation is a direct consequence of \Cref{clm:mif-connector}.
 
    \begin{obs}\label{clm:mif-consistency}
        If $G_{\ve{a}}$ and $G_{\ve{b}}$ are both full with $\loc(\ve{b}) = \loc(\ve{a})+\dir{e}_i$, and $G_{\ve{a}}$~encodes~$\ve{s}$, $G_{\ve{b}}$~encodes~$\ve{t}$, then $\ve{s}(i) = \ve{t}(i)$.
    \end{obs}
 
    \begin{clm}\label{clm:mif-nonfull}
        If $|F^\star| = \Gamma - \ell$, then the number of partial gadgets is at most $2\ell$.
    \end{clm}
    \begin{proof}
        By \Cref{obs:full-gadget-size}, every partial gadget causes a loss of
        at least one ball compared to a full gadget. If this lost ball is located on a connector of a partial gadget, then it makes two gadgets partial. Hence, if $t$ is the number of partial gadgets, then $|F^\star| \le \Gamma - t/2$. Since $|F^\star| = \Gamma - \ell$, it follows that $t \le 2\ell$.
    \end{proof}
 
    We define an assignment to the variables of $I$ based on $F^\star$ as follows. For each full gadget $G_{\ve{v}}$ encoding $\ve{s} \in C_{\ve{v}}$, set $f(\ve{v}) = \ve{s}$; for each partial gadget, assign $f(\ve{v})$ arbitrarily from $C_{\ve{v}}$. By \Cref{clm:mif-cyclic} and \Cref{clm:mif-consistency}, adjacent full gadgets satisfy their binary constraint. The only constraints that may be violated are those incident to partial gadgets. Each variable participates in at most $2d$ binary constraints and one unary constraint, so each partial gadget is incident to at most $2d + 1$ constraints. Since every assignment to $I$ violates at least $\delta\, \totalconstraints$ constraints, we obtain $(2d+1)t \ge \delta\,\totalconstraints$ and $t \ge \delta\,\totalconstraints/(2d+1)$. By \Cref{clm:mif-nonfull},
    \[
    |F^\star| \le \Gamma - \frac{\delta\, \totalconstraints}{2(2d+1)}.
    \]
    Since $\Gamma = \Theta(|V|)$ and $\totalconstraints = \Theta(|V|)$, there exists a constant $c'_d > 0$ such that $\frac{\delta\,\totalconstraints}{2(2d+1)} \ge c'_d\,\delta\, \Gamma$, and therefore $|F^\star| \le (1 - c'_d\,\delta)\,\Gamma$, Since this holds for every induced forest $F^\star$, we conclude $\OPT(\mathcal{B}) \le (1 - c'_d\,\delta)\,\Gamma$.

    \subparagraph{Running Time.} 
    Set $\eps = c'_d \delta / 2$. Assume toward a contradiction that there is a PTAS for \MIF{} of $n$ unit balls in $\R^d$ running in time $n^{\gamma_0/\eps^{d-1}}$ for some constant $\gamma_0 > 0$. Run this PTAS on $\mathcal{B}$ with parameter $\eps$.
 
    If $I$ is satisfiable, the optimal forest has size $\Gamma = (g + 2\tau)|V| + 12|E|$, so the PTAS returns a forest of size at least $(1-\eps)\Gamma$.
    If every assignment to $I$ violates $\ge \delta \totalconstraints$ constraints, every induced forest has size at most $(1-2\eps)\Gamma < (1-\eps)\Gamma$. Hence the PTAS distinguishes the two cases in time
    $|\mathcal{B}|^{\gamma_0/\eps^{d-1}} = \poly(|\constraints|,|D|)^{\gamma_0/\delta^{d-1}}$. Choosing $\gamma_0$ small enough contradicts \Cref{thm:pgcsp-to-tcsp}, completing the proof.
\end{proof}

For $d\ge3$, by replacing each unit ball with a unit cube in the construction of \Cref{thm:mif-ball} and following the same arguments we obtain an analogues result for unit cubes in $\R^d$. Note that our construction of a \MIF{} instance for unit balls in $\R^2$ does not generalize to \MIF{} on unit cubes in its current form due to space limitations in the construction.

\begin{cor} \label{thm:mif-cube}
    Let $d \ge 3$ be an integer. There is a real $\gamma > 0$ such that if the \MIF{} of $n$ unit cubes in $\R^d$ admits a PTAS with running time $n^{\gamma/\eps^{d-1}}$ then Gap-ETH fails.
\end{cor}

\subsection{\MIM{} Problem}

We now turn to the \MIM{} problem. Recall that \MIM{} for a set $\mathcal B$ of geometric objects asks for a maximum induced matching in the intersection graph of~$\mathcal B$. An induced matching is a set $M$ of edges such that the subgraph of $G$ induced by the endpoints of~$M$ is $1$-regular. We prove the equivalent of \Cref{thm:mif-ball} for the \MIM{} problem. The construction is identical to the \MIF{} construction, except that we remove the forest balls.

\begin{thm} \label{thm:mim-ball}
    Let $d \ge 2$ be an integer. There is a real $\gamma > 0$ such that if \MIM{} on intersection graphs of $n$ unit balls in $\R^d$ admits a PTAS with running time $n^{\gamma/\eps^{d-1}}$ then Gap-ETH fails.
\end{thm}

\begin{proof}
    Let $d \ge 2$ be a fixed integer. We reduce from \tcsp{} (\Cref{thm:pgcsp-to-tcsp}). Let $I=(V,\domain,\constraints)$ be an instance of $d$-dimensional \tcsp{} with $D=[\Lambda]^d$ and primal graph $G=(V,E)$ isomorphic to an induced subgraph of $G_k^d$. Let $\totalconstraints$ be the total number of constraints. We construct a set $\mathcal{B}$ of unit balls in $\R^d$ such that the maximum induced matching in the intersection graph of~$\mathcal{B}$ encodes the optimal value of~$I$. We denote by $\OPT(\mathcal B)$ the size of a maximum induced matching in the intersection graph of~$\mathcal B$.
 
    \subparagraph{Construction.}
    The construction is obtained from the \MIF{} construction of \Cref{thm:mif-ball} by \emph{removing all forest balls} (i.e., the entire global forest~$\mathcal{F}$). We retain the \emph{clique balls} at each clique position (in gadgets and connectors) with one private ball per clique position
    The geometric placement of all retained balls is identical to the \MIF{} construction. Each variable $\ve{v} \in V$ has a gadget $G_{\ve{v}}$ organized around a box $Q_{\ve{v}}$ with $\tau = O(d)$ clique positions arranged in a Hamiltonian cycle on the boundary, and adjacent gadgets are connected via entry and exit chains with $2$ intermediate clique positions each. At each clique position~$p_j$, the private ball $h_j(\ve{v})$ intersects every clique ball at $p_j$ and no other ball in~$\mathcal{B}$. The total number of balls is polynomial in the input size~$|\mathcal{B}| = \poly(|\constraints|,|\domain|)$.

    We list the key structural observations. The proofs are almost identical to those in \MIF{}.
    
    \begin{clm}\label{clm:mim-structure}
        Any induced matching $M$ in the intersection graph of~$\mathcal{B}$ uses at most one edge incident to any clique position.
    \end{clm}
 
    \begin{proof}
        Suppose $M$ contains two edges involving balls at position~$p_j$. Since the private ball $h_j(\ve{v})$ has degree~$1$ in~$M$ (it can be matched to at most one ball), at least one edge of~$M$ at~$p_j$ must be between two clique balls $B_j(\ve{v},\ve{s})$ and $B_j(\ve{v},\ve{s}')$. But both of these intersect $h_j(\ve{v})$, and if $h_j(\ve{v})$ is also an endpoint of another edge in~$M$, then $h_j(\ve{v})$ is adjacent to an endpoint of a different edge in~$M$, violating the induced matching property. If $h_j(\ve{v})$ is not in~$M$ at all, then the edge $\{B_j(\ve{v},\ve{s}), B_j(\ve{v},\ve{s}')\}$ is in $M$, and any second edge at~$p_j$ would share a vertex or be adjacent, again a contradiction. Hence at most one matching edge can use any clique position.
    \end{proof}

    \begin{clm}[Cyclic ordering]\label{clm:mim-cyclic}
        Let $M$ be an induced matching in which positions $p_j$ and $p_{j+1}$ (consecutive along the Hamiltonian cycle of a gadget $G_{\ve{v}}$) each contribute an edge: one incident to $B_j(\ve{v},\ve{s})$ and the other incident to $B_{j+1}(\ve{v},\ve{s}')$. Then $B_j(\ve{v},\ve{s})$ and $B_{j+1}(\ve{v},\ve{s}')$ are non-adjacent, and $\iota(\ve{s}) \le \iota(\ve{s}')$.
    \end{clm}

    \begin{proof}
        Since $M$ is an induced matching, no edge of the intersection graph connects an endpoint of one matching edge to an endpoint of another. In particular, $B_j(\ve{v},\ve{s})$ and $B_{j+1}(\ve{v},\ve{s}')$ must be non-adjacent. By \Cref{clm:mif-cyclic}, this is equivalent to $\iota(\ve{s}) \le \iota(\ve{s}')$.
    \end{proof}

    Applying \Cref{clm:mim-cyclic} around the entire cycle of $G_{\ve{v}}$ yields $\iota(\ve{s}) \le \iota(\ve{s}') \le \cdots \le \iota(\ve{s})$, forcing all matched clique balls within a single gadget to encode the same domain value. The same argument applied to the connector chains (using \Cref{clm:mif-connector}) enforces equality of the relevant coordinate across adjacent gadgets, exactly as in the \MIF{} proof.

    \begin{clm}[Connector consistency]\label{clm:mim-connector}
        Let $M$ be an induced matching in which every clique position in the entry and exit chains between adjacent gadgets $G_{\ve{a}}$ and $G_{\ve{b}}$ ($\loc(\ve{b}) = \loc(\ve{a})+\dir{e}_i$) contributes an edge, where $G_{\ve a}$ encodes $\ve s$ and $G_{\ve b}$ encodes $\ve t$. Then $\ve{s}(i) = \ve{t}(i)$.
    \end{clm}
 
    \begin{proof}
        Along the entry chain, consecutive matched clique balls must be non-adjacent (by the induced matching property). By \Cref{clm:mif-connector}, this enforces $\ve{s}(i) \ge q_1 \ge q_2 \ge \ve{t}(i)$. Along the exit chain, the same argument gives $\ve{s}(i) \le \ve{t}(i)$. Together, $\ve{s}(i) = \ve{t}(i)$.
    \end{proof}

    \subparagraph{Exact Equivalence.}
    We show that $\OPT(\mathcal{B}) = \tau|V| + 4|E|$ if and only if $I$ is satisfiable. First, note that by \Cref{clm:mim-structure}, any induced matching contains at most one edge per clique position. The total number of clique positions is $\tau|V| + 4|E|$ ($\tau$ per gadget and $4$ per edge of~$G$). Hence $\OPT(\mathcal{B}) \le \tau|V| + 4|E|$. Assume $I$ has a solution $f : V \to [\Lambda]^d$ with $f(\ve{v}) \in C_{\ve{v}}$ for all~$\ve{v}$. We construct an induced matching $M$ as follows: for each clique position~$p_j$ (in gadgets and connectors), include the edge $\{h_j(\ve{v}), B_j(\ve{v}, f(\ve{v}))\}$, where $f(\ve{v})$ determines the selected clique ball.
 
    Each ball appears in at most one edge by \Cref{clm:mim-structure}. It remains to check that no edge of the intersection graph connects endpoints of distinct matching edges. Since $f$ is satisfying, the cyclic ordering is respected: within each gadget, consecutive matched clique balls correspond to the same value $\iota(f(\ve{v}))$ and are non-adjacent by \Cref{clm:mim-cyclic}. Across each connector in direction~$\dir{e}_i$, consecutive matched clique balls encode the same $i$-th coordinate value and are non-adjacent by \Cref{clm:mim-connector}. Private balls at different positions do not intersect each other and do not intersect clique balls at other positions. Hence $M$ is an induced matching of size $|M| = \tau|V| + 4|E|$.

    Assume $\OPT(\mathcal{B}) = \tau|V| + 4|E|$ and let $M$ be an induced matching of this size. Since the upper bound is $\tau|V| + 4|E|$ and $|M|$ achieves it, every clique position contributes exactly one edge to~$M$. By \Cref{clm:mim-structure}, each such edge is of the form $\{h_j(\ve{v}), B_j(\ve{v},\ve{s}_j)\}$ for some $\ve{s}_j \in C_{\ve{v}}$. Within each gadget $G_{\ve{v}}$, every pair of consecutive clique positions along the Hamiltonian cycle contributes an edge. By \Cref{clm:mim-cyclic}, all matched clique balls within $G_{\ve{v}}$ correspond to the same domain value $\ve{s}_{\ve{v}} \in C_{\ve{v}}$. Define $\bar f(\ve{v}) := \ve{s}_{\ve{v}}$. 

    For each edge $\{\ve{a},\ve{b}\} \in E$ with $\loc(\ve{b}) = \loc(\ve{a})+\dir{e}_i$, every position in the entry and exit chains contributes an edge. Therefore, \Cref{clm:mim-connector} implies that the binary constraints between adjacent variables are satisfied. The unary constraints are satisfied since $\ve{s}_{\ve{v}} \in C_{\ve{v}}$ by construction. Hence $\bar f$ is a solution for~$I$.

    \subparagraph{Error Analysis.}
    Assume every assignment to $I$ violates at least $\delta \totalconstraints$ constraints. We show that every induced matching in the constructed graph has size at most $(1 - c'_d\,\delta)(\tau|V| + 4|E|)$ for a constant $c'_d > 0$.
 
    Let $M$ be a maximum induced matching. We call a gadget $G_{\ve{v}}$ \emph{full} if $M$ contains an edge at every clique position of~$G_{\ve{v}}$ and all incident connectors, and all matched clique balls within the gadget correspond to a single domain value $\ve{s} \in C_{\ve{v}}$, with connector intermediates encoding $\ve{s}(i)$ consistently. Each full gadget contributes $\tau + 4|E_{\ve{v}}|$ edges to $M$ (where $E_{\ve{v}}$ is the set of edges
    incident to $\ve{v}$), and each partial gadget contributes strictly fewer.
 
    By \Cref{clm:mim-connector}, adjacent full gadgets satisfy their binary constraint. Only constraints incident to partial gadgets may be violated. Each partial gadget is incident to at most $2d+1$ constraints. If $t$ is the number of partial gadgets, then $(2d+1)t \ge \delta \totalconstraints$, giving $t \ge \delta \totalconstraints/(2d+1)$.
 
    Each partial gadget causes a loss of at least one edge compared to a full gadget (a lost edge on a connector may be shared between two partial gadgets, so the loss is at least $t/2$). Hence
    \[
        |M| \le \tau|V| + 4|E| - \frac{\delta \totalconstraints}{2(2d+1)}.
    \]
    Since $\tau|V| + 4|E| = \Theta(|V|)$ and $\totalconstraints = \Theta(|V|)$, there exists $c'_d > 0$ such that $|M| \le (1 - c'_d\,\delta)(\tau|V| + 4|E|)$.

    \subparagraph{Running Time.}
    The running-time argument is identical to that in the proof of \Cref{thm:mif-ball}. Set $\eps = c'_d\delta/2$. A PTAS for \MIM{} running in time $n^{\gamma_0/\eps^{d-1}}$ would distinguish the two cases with running time $\poly(|\constraints|,|\domain|)^{O(1)/\delta^{d-1}}$, contradicting \Cref{thm:pgcsp-to-tcsp} for sufficiently small~$\gamma_0$.
\end{proof}

We extent our construction to unit cubes. Note that unlike in the \MIF{} reduction, forest balls are not present. For the reduction using unit cubes, one can alternate the placement of \emph{private cubes} in the exterior and interior of $Q_{\ve{v}}$ for the adjacent clique positions so that the private cubes remain pairwise disjoint. Therefore, the construction also works for unit squares in $\R^2$. We state this in the following corollary.
 
\begin{cor} \label{thm:mim-cube}
    Let $d \ge 2$ be an integer. There is a real $\gamma > 0$ such that if the \MIM{} on intersection graphs of $n$ unit cubes in $\R^d$ admits a PTAS with running time $n^{\gamma/\eps^{d-1}}$ then Gap-ETH fails.
\end{cor}

\subsection{\MDS{} Problem}

We prove a PTAS lower bound for the \MDS{} problem on unit balls and unit cubes in~$\R^d$, adapting the gadget construction of Marx~\cite{Marx06MDS} to arbitrary dimension~$d \ge 2$. In dimension $d=2$, Marx's construction uses a gadget organized around an axis-aligned square with $16$ blocks $X_1, Y_1, X_2, Y_2, \dots, X_8, Y_8$ alternating in cyclic order on its boundary; the cyclic domination ordering of consecutive $X$-blocks through the intermediate $Y$-blocks forces each gadget to represent a single value from the domain. We generalize this to $d \ge 2$ by replacing the square with a $d$-dimensional box and using the Hamiltonian cycle on the cross-polytope from the proof of the \textsc{Maximum Induced Forest} lower bound (\Cref{thm:mif-ball}) to arrange the blocks cyclically on the boundary of the box. In the following proof we assume that the reader is familiar with the construction of~\cite{Marx06MDS}.
 
\begin{thm} \label{thm:mds-ball}
    Let $d \ge 2$ be an integer. There is a real $\gamma > 0$ such that if \MDS{} on intersection graphs of $n$ unit balls or unit cubes in $\R^d$ admits a PTAS with running time $n^{\gamma/\varepsilon^{d-1}}$ then Gap-ETH~fails.
\end{thm}
 
\begin{proof}
The proof is written for unit balls. The same argument applies to unit cubes after replacing the local ball gadgets by the corresponding cube gadgets. Let $d \ge 2$ be a fixed integer. We reduce from $d$-dimensional \tcsp{} (\Cref{thm:pgcsp-to-tcsp}). Let $I = (V, D, \constraints)$ be an instance of $d$-dimensional \tcsp{} with $D = [\Lambda]^d$. The primal graph $G = (V,E)$ of $I$ is isomorphic to an induced subgraph of $G_k^d$ via a subgraph isomorphism $\loc : V\rightarrow [k]^d$. Let $\totalconstraints$ be the total number of constraints. We construct a set $\mathcal{B}$ of unit balls in $\R^d$ such that the minimum dominating set of the intersection graph of $\mathcal{B}$ encodes the optimal value of~$I$. We denote the size of the minimum dominating set of the intersection graph of~$\mathcal{B}$ by~$\OPT(\mathcal{B})$. 
 
\subparagraph{Construction.}
Let $\rho = 1/(400\,\Lambda^{d+2})$ and let $\iota : [\Lambda]^d \to \{0,1,\dots,\Lambda^d - 1\}$ be an arbitrary bijection. For each variable $\ve{v} \in V$, we construct a \emph{variable gadget} $G_{\ve{v}}$ organized around an axis-aligned $d$-dimensional box $Q_{\ve{v}}$ of side length $8$, centered at $c_{\ve{v}} := 10 \, \loc(\ve{v})$.

\noindent\emph{Hamiltonian cycle on the boundary.}
As in the proof of \Cref{thm:mif-ball}, we arrange block positions on the boundary of $Q_{\ve{v}}$ using a Hamiltonian cycle of the cross-polytope. On each facet $F_i$, we place five consecutively adjacent block positions (similar to the $2$-dimensional case), including designated entry and exit port blocks (the second block and the fourth block, respectively). Transition subpaths between consecutive facets are routed along the boundary of $Q_{\ve{v}}$, yielding a directed Hamiltonian cycle $H_{\ve{v}}$ through $2\tau = O(d)$ block positions.
 
\noindent\emph{Alternating $X$- and $Y$-blocks.}
Following Marx~\cite{Marx06MDS}, we place alternating $X$- and $Y$-blocks around the cycle in the order $X_1, Y_1, X_2, Y_2, \dots, X_\tau, Y_\tau$. Each $Y$-block $Y_j$ sits between $X$-blocks $X_j$ and $X_{j+1}$ (indices modulo~$\tau$). Note that all the port blocks are of type $X$. 
Each $X$-block $X_j$ contains $|C_\ve{v}|$ balls (one per domain element $\ve{s} \in C_{\ve{v}}$), and each $Y$-block $Y_j$ contains $\Lambda^d + 1$ balls.
 
The balls within each block are placed with small perturbations of magnitude $O(\rho)$ along the tangent direction of $H_{\ve{v}}$ at the block position. At port positions, additional perturbations of magnitude $O(\rho)$ in the facet normal direction encode the relevant coordinate value, similar to the placement introduced in Theorem~1 of~\cite{Marx06MDS} and the clique positions of the \MIF{} construction (\Cref{thm:mif-ball}). Two balls can intersect only if they belong to the same block or to adjacent blocks along the cycle. The crucial property, established in~\cite[Lemma~1]{Marx06MDS} for $d=2$ and extending to $d \ge 2$ via our Hamiltonian cycle arrangement, is the following.
 
\begin{clm}[Cyclic domination; cf.\ {\cite[Lemma~1]{Marx06MDS}}]\label{clm:mds-cyclic}
Let $X_{j, r}$ be a unit ball in the block $X_j$ that encodes the value $\ve{s} \in C_\ve{v}$, where $r = \iota(\ve{s})$. Two balls $X_{j, r} \in X_j$ and $X_{j+1, r'} \in X_{j+1}$ together dominate every ball of $Y_{j+1}$ if and only if $r \ge r'$.
\end{clm}

The ball $X_{j, r_1}$ dominates exactly the prefix $Y_{j+1, 0}, \dots, Y_{j+1, r_1 - 1}$ from block~$Y_{j+1}$, and the ball $X_{j+1, r_2}$ dominates exactly the suffix $Y_{j+1, r_2}, \dots, Y_{j+1, \Lambda^d}$. These two ranges cover all of $Y_{j+1}$ if and only if $r_1 \ge r_2$. As in~\cite[Lemma~1]{Marx06MDS}, suppose a dominating set contains exactly one ball from each $X$-block of gadget $G_{\ve{v}}$ and no ball from any $Y$-block. The $Y$-blocks cannot be entirely dominated by balls from a single neighboring $X$-block, so both adjacent $X$-blocks must contribute. The cyclic domination ordering around $H_{\ve{v}}$ then yields $r_1 \ge r_2 \ge \cdots \ge r_\tau \ge r_1$, forcing $r_1 = r_2 = \cdots = r_\tau$. Thus each gadget represents a single domain value $\ve{s}_{\ve{v}} \in C_{\ve{v}}$ in any optimal dominating set.

\noindent\emph{Connectors.}
For each edge $\{\ve{a}, \ve{b}\} \in E$ with $\loc(\ve{b}) = \loc(\ve{a}) + \dir{e}_i$, we connect gadgets $G_{\ve{a}}$ and $G_{\ve{b}}$ through the aligned port positions on facets $F_i$ and $F_{i+d}$. Following~\cite{Marx06MDS}, we add entry and exit connector blocks between the port positions (each containing $\Lambda + 1$ balls). The connector enforces that if gadget $G_{\ve{a}}$ represents $\ve{s}$ and gadget $G_{\ve{b}}$ represents $\ve{t}$, then the dominating set covers all connector balls only if $\ve{s}(i) = \ve{t}(i)$. The entry block enforces $\ve{s}(i) \ge \ve{t}(i)$ and the exit block enforces $\ve{s}(i) \le \ve{t}(i)$.

\subparagraph{Exact Equivalence.}
At least $\tau$ balls are required to dominate the $Y$-blocks of each gadget. The $Y$-blocks cannot all be dominated by balls from a single $X$-block. Each of the $\tau$ $Y$-blocks requires coverage from both of its neighboring $X$-blocks, forcing at least one selected ball per $X$-block. Hence every dominating set of size $\tau |V|$ contains exactly one ball from each $X$-block and no ball from any $Y$-block.
 
If $I$ has a solution $f : V \to [\Lambda]^d$ with $f(\ve{v}) \in C_{\ve{v}}$ for all $\ve{v}$, we select from each gadget $G_{\ve{v}}$ the $\tau$ balls $X_{1, \iota(f(\ve{v}))}, \dots, X_{\tau, \iota(f(\ve{v}))}$. By the cyclic domination property \Cref{clm:mds-cyclic}, these dominate all $Y$-blocks within the gadget. Since $f$ satisfies every binary equality constraint all connector blocks are also dominated. Hence $\OPT(\mathcal{B}) = \tau |V|$.
 
Conversely, if $\OPT(\mathcal{B}) = \tau|V|$, the cyclic ordering forces each gadget to encode a single value $\ve{s}_{\ve{v}} \in C_{\ve{v}}$, and the connectors force $\ve{s}_{\ve{a}}(i) = \ve{s}_{\ve{b}}(i)$ for all adjacent pairs. Hence the assignment $f(\ve{v}) = \ve{s}_{\ve{v}}$ satisfies $I$. The detailed verification follows from~\cite[Theorem~1]{Marx06MDS} applied to our generalized gadget arrangement.

\subparagraph{Error Analysis.}
Assume that every assignment to $I$ violates at least $\delta\,\totalconstraints$ constraints. Let $\mathfrak{D}^\star$ be a minimum dominating set of the intersection graph of~$\mathcal{B}$. We call a gadget $G_{\ve{v}}$ \emph{canonical} if $\mathfrak{D}^\star$ contains exactly one ball from each $X$-block of $G_{\ve{v}}$, all corresponding to the same domain value $\ve{s} \in C_{\ve{v}}$, with the assigned connector blocks encoding $\ve s(i)$ consistently. Each canonical gadget contributes exactly $\tau$ balls. Each non-canonical gadget requires at least $\tau + 1$ balls to dominate its $Y$-blocks: if the selected $X$-block representatives are not all equal, some $Y$-block between two consecutive $X$-blocks with different values remains partially undominated and requires an additional ball.
 
From $\mathfrak{D}^\star$, define an assignment $f$ to $I$. For each canonical gadget encoding $\ve{s}$, set $f(\ve{v}) = \ve{s}$; for non-canonical gadgets, assign $f(\ve{v})$ arbitrarily from $D$. By the connector structure, adjacent canonical gadgets satisfy their binary constraint. The only constraints that may be violated are those incident to non-canonical gadgets. Each variable is incident to at most $2d + 1$ constraints (one unary and $2d$ binary). Let $t$ denote the number of non-canonical gadgets. Since every assignment violates at least $\delta \totalconstraints$ constraints, we have $(2d+1)t \ge \delta \totalconstraints$, giving $t \ge \delta \totalconstraints / (2d+1)$.
 
Each non-canonical gadget causes at least one extra ball beyond the $\tau$-ball budget. An extra ball on a shared connector may be counted by two non-canonical gadgets, so
\[
|\mathfrak{D}^\star| \ge \tau |V| + \frac{t}{2} \ge \tau|V| + \frac{\delta \totalconstraints}{2(2d+1)}.
\]
Since $\tau|V| = \Theta(|V|)$ and $\totalconstraints = \Theta(|V|)$, there exists a constant $c'_d > 0$ such that
\[
\OPT(\mathcal{B}) \ge (1 + c'_d \,\delta)\,\tau|V|.
\]
 
\subparagraph{Running Time.}
Set $\varepsilon = c'_d \delta / 2$. Assume toward a contradiction that there is a PTAS for \MDS{} of $n$ unit balls in $\R^d$ running in time $n^{\gamma_0 / \varepsilon^{d-1}}$ for some constant $\gamma_0 > 0$. Run this PTAS on $\mathcal{B}$ with parameter~$\varepsilon$.
 
If $I$ is satisfiable, $\OPT(\mathcal{B}) = \tau|V|$, and the PTAS returns a dominating set of size at most $(1 + \varepsilon)\tau|V|$. If every assignment to $I$ violates at least $\delta \totalconstraints$ constraints, then $\OPT(\mathcal{B}) \ge (1 + 2\varepsilon)\tau|V| > (1+\varepsilon)\tau|V|$, so any feasible dominating set has size exceeding $(1+\varepsilon)\tau|V|$. Hence the PTAS distinguishes the two cases. The running time is
\[
|\mathcal{B}|^{\gamma_0/\varepsilon^{d-1}} = \mathrm{poly}(|\constraints|, |D|)^{O(1)/\delta^{d-1}}.
\]
Choosing $\gamma_0$ small enough contradicts \Cref{thm:pgcsp-to-tcsp}, completing the proof.
\end{proof}

\subsection{\MPS{} Problem}

In this section, we prove a lower bound for the running time of a PTAS for \MPS{} on unit balls in $\R^d$. We use the reduction described by Blank~\etal~\cite{geert} as a black box, which proves a running-time lower bound for exact algorithms to solve the $k$-center clustering problem. They reduce a variant of the constraint satisfaction problem called \bsumset{} to the $k$-center problem. For a positive integer $\kappa$, an instance $J = (V, D, \bar \constraints)$ of \bsumset{} consists of a primal graph $G = (V,E)$ isomorphic to an induced subgraph of a $d$-dimensional grid graph $G_\kappa^d$, domain $D = [\mathfrak{n}^2]_\circ$, and a family $\bar \constraints$ of subsets of $D$. The constraints in $\bar \constraints$ consist of a unary constraint $\bar C_\ve{v} \subseteq D$ for each vertex $\ve v \in V$ and a binary constraint $\bar C_e \subseteq D$ for each edge $e \in E$. The graph $G$ is properly $2$-colored so that $\bar C_\ve{v} \subseteq [\mathfrak{n}]_\circ \cdot \mathfrak{n}^{\chi(\ve{v})}$ where $\chi(\ve v) \in \{0,1\}$ is the color of $\ve v$, and $\bar C_e \subseteq \{a+b : (a,b) \in \bar C_\ve{v} \times \bar C_\ve{u}\}$ for $e = \{\ve u,\ve v\}$. An assignment $f: V \to D$ of $J$ is a solution if $f(\ve v) \in \bar C_\ve{v}$ for all $\ve{v} \in V$ and $f(\ve u)+f(\ve v) \in \bar C_e$ for every edge $e = \{\ve u,\ve v\} \in E$.

Recall that a collection of unit balls $\cB$ is a \emph{cover} of a point set $P$ if each point of $P$ is contained in at least one ball of $\cB$. We use the following consequence of the construction of Blank~\etal~\cite[Sections~5.1--5.3]{geert}.\footnote{The reduction in~\cite{geert} constructs a $k$-center instance with balls of radius $r_0 = 1-\varepsilon_0+\varepsilon_0^{1.7}$ for a small constant $\varepsilon_0 = (10\mathfrak{n})^{-100}$. After uniformly scaling the point set by a factor $1/2r_0$, we may assume all balls have unit diameter.}

\begin{thm}
\label{thm:bsumset-to-cover}
For every fixed integer $d\ge 2$, there exist constants $h_d,q_d>0$ and a polynomial-time construction with the following properties. Given a $d$-dimensional \bsumset{} instance $J=(V,[\mathfrak n^2]_\circ,\bar\constraints)$ with primal graph $G=(V,E)$, the construction outputs a point set $P_J\subseteq \R^d$ and an integer $K \le h_d |V|$ such that $|P_J|=\poly(|V|,\mathfrak n)$ and 
\begin{enumerate}
    \item[(i)] every cover of $P_J$ by unit balls has size at least $K$;

    \item[(ii)] if $J$ is satisfiable, then $P_J$ can be covered by exactly $K$ unit balls;

    \item[(iii)] for every cover of $P_J$ by $K+t$ unit balls, one can extract an assignment $f:V\to[\mathfrak n^2]_\circ$ for $J$ that violates at most $q_d t$ constraints of $J$.
\end{enumerate}
\end{thm}

\noindent\textit{Remark.} In the terminology of~\cite{geert}, the value $K$ is the number of canonical balls forced by the anchor points. The assignment is read from the splitting indices along the basic curves. The final property follows from the locality of the construction: a non-canonical ball can interact with only constantly many unary or binary constraint gadgets, where the constant depends only on $d$. 

\medskip
We now state the main theorem of this subsection.

\begin{thm} \label{thm:piercing}
    Let $d \ge 2$ be an integer. There is a real $\gamma > 0$ such that if \MPS{} on $n$ unit balls in $\R^d$ admits a PTAS with running time $n^{\gamma/\eps^{d-1}}$ then Gap-ETH fails.
\end{thm}

\begin{proof}
    For convenience we view \MPS{} as the equivalent \emph{covering} formulation: given a finite point set $P\subseteq \R^d$, find a minimum-cardinality set of unit balls whose union covers $P$. For a given point set $P$ we denote by $\OPT(P)$ the minimum number of unit balls that cover $P$.

    \subparagraph{Construction.} Let $I=(V,D,\constraints)$ be a $d$-dimensional \lecsp{}. Let $G=(V,E)$ be the primal graph of $I$, and let $D=[\Lambda]^d$. The primal graph of $I$ is isomorphic via $\loc : V \rightarrow V(G_k^d)$ to an induced subgraph of the $d$-dimensional grid graph of side length $k$. Let $\totalconstraints$ be the total number of constraints of~$I$. Recall that for each $i \in [d]$, in a \lecsp{} instance each binary constraint on a pair of adjacent variables $\ve{u}$ and $\ve{v}$ with $\loc(v) = \loc(u) + \dir{e}_i$ enforces the $\le$-inequality on coordinate~$i$. We define an equivalent \bsumset{} instance $J = (V, D = [\mathfrak{n}^2]_\circ,\bar\constraints)$ with the same primal graph as $I$ and the same subgraph isomorphism into~$G_k^d$. Let $\mathfrak{n} = (2\Lambda+1)^d$. For each $\ve{v} \in V$, define the coloring $\chi(\ve{v}) := (\sum_{j=1}^d \tilde{\ve{v}}(j))\bmod 2$, where $\tilde{\ve{v}} = \loc(\ve{v})$. For a tuple $\ve{a}=(a_1,\dots,a_d)\in[\Lambda]^d$, define
    \[
    \iota_{\ve{v}}(\ve{a}) := \sum_{j=1}^d a_j (2\Lambda+1)^{j-1}.
    \]
    For each variable $\ve{u} \in V$ define the allowed integer set
    \[
    \bar C_{\ve{u}} := \{\iota_{\ve{u}}(\ve{a}) \cdot \mathfrak{n}^{\chi(\ve{u})} : \ve{a}\in C_{\ve{u}}\}.
    \]
    For each $i\in[d]$, if variables $\ve u$ and $\ve v$ are adjacent with $\loc(v) = \loc(u) + \dir{e}_i$, define
    \[
    \bar C_{\{\ve u, \ve v\}} := \Bigl\{\iota(x) \cdot \mathfrak{n}^{\chi(u)} + \iota(y) \cdot \mathfrak{n}^{\chi(v)}:(x,y) \in [\Lambda]^d \times [\Lambda]^d,\, x_i \le y_i\Bigr\}.
    \]
    For each vertex $\ve u$, the set $\bar C_{\ve u}$ represents the unary constraint, and each edge $\{\ve{u},\ve{v}\} \in E$ with $\loc(\ve{v}) = \loc(\ve{u}) + \dir{e}_i$ has a binary constraint $\bar C_{\{\ve u, \ve v\}}$.

    \begin{lem}\label{lem:lecsp-to-bsumset}
    Let $\ve{u}$ and $\ve{v}$ be two variables with $\loc(\ve{v}) = \loc(\ve{u}) + \dir{e}_i$. For any $\ve{x} =(x_1,\dots,x_d)\in[\Lambda]^d$ and $\ve{y}=(y_1,\dots,y_d)\in[\Lambda]^d$, $\ve x_i\le \ve y_i$ if and only if $\iota_{\ve{u}}(\ve{x})+\iota_{\ve{v}}(\ve{y})\in \bar C_{\{\ve u, \ve v\}}$.
    \end{lem}

    \begin{proof}
        Since $\chi(\ve u)\neq \chi(\ve v)$, the two terms
        \[
        \iota(\ve x)\mathfrak n^{\chi(\ve u)}
        \quad\text{and}\quad
        \iota(\ve y)\mathfrak n^{\chi(\ve v)}
        \]
        occupy disjoint base-$\mathfrak{n}$ digit blocks. As $\iota(\ve z)<\mathfrak{n}$ for all $\ve z\in[\Lambda]^d$, their sum is carry-free and uniquely determines the pair $(\ve x,\ve y)$. By definition,
        \[
        \bar C_{\{\ve u,\ve v\}} = \Bigl\{ \iota(\ve x')\mathfrak n^{\chi(\ve u)} + \iota(\ve y')\mathfrak n^{\chi(\ve v)} : x'_i \le y'_i \Bigr\}.
        \]
        Thus,
        \[
        \iota(\ve x)\mathfrak n^{\chi(\ve u)} + \iota(\ve y)\mathfrak n^{\chi(\ve v)}\in \bar C_{\{\ve u,\ve v\}}
        \]
        holds if and only if $(\ve x,\ve y)$ satisfies $x_i\le y_i$, since the representation is unique.
    \end{proof}

    Note that \Cref{lem:lecsp-to-bsumset} implies that $I$ is satisfiable if and only if $J$ is satisfiable. Moreover, the reduction between $I$ and $J$ preserves constraint satisfaction constraint-by-constraint, so if every assignment violates at least $\delta \totalconstraints$ constraints in $I$, then every assignment violates at least $\delta \totalconstraints$ constraints in $J$.

    Finally, apply the construction of~\cite{geert} to $J$ and obtain a point set $P_J \subseteq \R^d$ and the value $K$ as stated in \Cref{thm:bsumset-to-cover}.

    \subparagraph{Exact Equivalence.} Assume that $I$ is satisfiable. By \Cref{lem:lecsp-to-bsumset}, the instance $J$ is satisfiable. Then by \Cref{thm:bsumset-to-cover}, the point set $P_J$ can be covered by $K$ unit balls. Since every cover of $P_J$ has size at least $K$, we get $\OPT(P_J)=K$.

    Conversely, suppose that $P_J$ can be covered by exactly $K$ unit balls. By \Cref{thm:bsumset-to-cover}, this cover induces an assignment to $J$ violating at most $q_d\cdot 0=0$ constraints. Hence $J$ is satisfiable, and by \Cref{lem:lecsp-to-bsumset}, the original instance $I$ is satisfiable.

    \subparagraph{Error Analysis.}
    Now assume that every assignment to $I$ violates at least $\delta\totalconstraints$ constraints. By \Cref{lem:lecsp-to-bsumset}, every assignment of $J$ also violates at least $\delta\totalconstraints$ constraints. Let $\mathcal U$ be an arbitrary cover of $P_J$ by unit balls, and write $|\mathcal U|=K+t$. By \Cref{thm:bsumset-to-cover}, the cover $\mathcal U$ induces an assignment to $J$ that violates at most $q_d t$ constraints. Since every assignment of $J$ violates at least $\delta\totalconstraints$ constraints, we have $q_d t\ge \delta\totalconstraints$. Therefore every cover of $P_J$ has size at least $K+\frac{\delta\totalconstraints}{q_d}$. Using $\totalconstraints\ge |V|$ and $K\le h_d |V|$, we obtain
    \[
        \OPT(P_J)
        \ge
        K+\frac{\delta |V|}{q_d}
        \ge
        \left(1+\frac{\delta}{h_dq_d}\right)K.
    \]
     Set $c_d := 1/h_dq_d$. Run the assumed PTAS on $P$ with approximation parameter $\eps = c_d\delta/3$. If $I$ is satisfiable then $\OPT(P)=K$, so the PTAS returns a cover of size at most $(1+\eps)K$. If every assignment violates at least a $\delta$-fraction of constraints, then $\OPT(P)\ge (1+c_d\delta)K>(1+\eps)K$, and since any output cover has size at least $\OPT(P)$, the PTAS must return strictly more than $(1+\eps)K$ balls. Hence the PTAS distinguishes the two cases.

    \subparagraph{Running Time.} Finally, the construction in~\cite{geert} produces $n=|P|=\poly(|\constraints|,|D|)$ points, i.e. polynomial in the input size of $I$. Assume for contradiction that there is a PTAS for this problem running in time $n^{\gamma_0/\eps^{d-1}}$ on instances with $n$ points, for some constant $\gamma_0>0$. Therefore the running time of this PTAS on $P$ is
    \[
        n^{\gamma_0/\eps^{d-1}}
        = \poly(| \constraints|, |\domain|)^{\gamma_0/\delta^{d-1}},
    \]
    which contradicts \Cref{thm:pgcsp-to-lecsp} for a suitable constant $\gamma_0>0$.
\end{proof}

\section{Shifting Based Polynomial Time Approximation Schemes} \label{sec:algorithms}

A matching upper bounds for \MIS{} and \MDS{} on intersection graphs of~$n$ unit balls (or unit cubes) in~$\R^d$, with running time $n^{O(1/\eps^{d-1})}$, was established by~\cite{Chan03} and~\cite{DeDCN13,JalluPD13}, respectively. In this section, we apply the shifting technique to obtain the same running time for \MIF{}, \MIM{}, and \MPS{} on $n$ unit balls (or unit cubes) in~$\mathbb{R}^d$.

\subsection{PTAS for \MIF{}}

In this subsection, we match our lower bounds for \MIF{} and \MIM{} on intersection graph of unit balls (unit cubes) in $\IR^d$ by presenting a PTAS with a running time $n^{O(1/\eps^{d-1})}$. Our PTAS uses the shifting technique of Hochbaum and Maass~\cite{HochbaumM85}.

\begin{thm}\label{thm:mif-algo}
For every fixed $d \geq 2$ and $\eps > 0$, there is a $(1-\eps)$-approximation algorithm for \MIF{} on unit ball (unit cube) intersection graphs in $\R^d$ running in time $n^{O(1/\eps^{d-1})}$.
\end{thm}

\begin{proof}
We write the proof for the intersection graph of unit balls. The algorithm and its analysis for the unit cubes are similar. Fix $\eps > 0$. Set the shifting parameter $\ell = \lceil 2d / \eps \rceil$. For each $d$-tuple of offsets $i = (i_1, \ldots, i_d) \in \{0, \ldots, \ell - 1\}^d$, partition $\R^d$ into axis-aligned cubes of side length $\ell$:
\[
\text{Cell}_{\ve{j}}^{(\ve{i})} = \prod_{r=1}^{d} \big[i_r + j_r \ell,\; i_r + (j_r+1)\ell\big), \qquad \ve{j} \in \mathbb{Z}^d.
\]
A ball $B$ is a \emph{boundary ball} for offset $i$ if, for some coordinate $r$ and some integer $j_r$, the $r$-th coordinate of its center is within distance $1$ of the grid hyperplane $x_r = i_r + j_r \ell$.

The following is the algorithm to obtain a $(1-\eps)$-approximation for the \MIF{} problem with the set $\cB$ of unit balls as an input.

\begin{enumerate}
\item For each offset $i \in \{0, \ldots, \ell-1\}^d$:
  \begin{enumerate}
  \item Remove all boundary balls for offset $i$.
  \item The remaining balls partition into independent subproblems, one per grid cell. (After removing boundary balls, balls in different cells are non-adjacent, since any edge between cells would require at least one endpoint within distance $1$ of a cell boundary.)
  \item In each cell $C$, let $G_C$ denote the unit ball intersection graph on the (non-boundary) balls in $C$. Solve \MIF{} on $G_C$ exactly using \Cref{thm:mif-exact-k}.
  \item Let $F_{i}$ be the union of solutions across all cells.
  \end{enumerate}
\item Return the largest $F_{i}$ over all offsets.
\end{enumerate}

We show that this algorithm returns a $(1-\eps)$-approximate solution of $\OPT(\mathcal{B})$.

\begin{lem}\label{lem:shift}
    There exists an offset $i^\star$ such that the number of balls in an optimal induced forest $F^\star$ that are boundary balls for $i^\star$ is at most $\eps |F^\star|$.
\end{lem}
 
\begin{proof}
Fix a coordinate $r \in \{1, \ldots, d\}$.
For each offset $i_r \in \{0, \ldots, \ell-1\}$, a ball in $S^\star$ is a boundary ball in coordinate~$r$ for offset $i_r$ if and only if its center's $r$-th coordinate is within distance $1$ of some grid line $i_r + j_r \ell$.
Since the grid lines are spaced $\ell \geq 2$ apart, the interval of width~$2$ around any fixed center value meets at most $2$ grid lines, corresponding to at most $2$ offsets. Let $S^\star_{r,i_r}$ be the set of boundary balls in the $r$-th coordinate for the offset $i_r$.
Hence $
\sum_{i_r = 0}^{\ell - 1} |S^\star_{r, i_r}| \leq 2|S^\star|$. so $\min_{i_r} |S^\star_{r, i_r}| \leq 2|S^\star|/\ell$.
Taking a union bound over all $d$ coordinates, the best joint offset $i^\star$ loses at most
\[
d \cdot \frac{2|S^\star|}{\ell} \leq \eps \cdot |S^\star|
\]
balls from the optimum (by our choice $\ell = \lceil 2d/\eps \rceil$).
\end{proof}

Since subsets of induced forests are induced forests, $S^\star \setminus \{\text{boundary balls for } i^\star\}$ is a valid induced forest, and it decomposes independently across cells. Thus the algorithm returns a forest of size at least $(1-\eps)|\OPT|$.

Consider a fixed offset $i$ and a cell $C$ of side length $\ell$. Let $n_C$ be the number of non-boundary balls in $C$. Let $k_C^* := \OPT(G_C)$. We partition the cell (a cube of side $\ell$) into subcubes of side length $1/\sqrt{d}$. Each sub-cube has diameter~$1$, so all unit balls with centers in the same sub-cube are pairwise intersecting. In any induced forest, a clique contributes at most $2$ vertices. Hence the induced forest contains at most $2$ vertices per sub-cube. The number of sub-cubes is $(\ell\sqrt{d})^d = d^{d/2} \ell^d = O_d(\ell^d)$. Therefore $k_C^* \le d^{d/2}\,\ell^d = O(\ell^d)$.

Therefore by \Cref{thm:mif-exact-k}, solving \MIF{} on the non-boundary balls of cell $C$ of side length $\ell$ takes $n_C^{O(\ell^{d-1})}$ time. Summing over all cells, we have
\[
\sum_C n_C^{O(\ell^{d-1})} \le \Bigl(\max_C n_C\Bigr)^{O(\ell^{d-1})} \cdot |\{C : n_C > 0\}| \le n^{O(\ell^{d-1})} \cdot n = n^{O(\ell^{d-1})}.
\]
Over all $\ell^d = O(1/\eps^d)$ offsets, the total running time is $\ell^d \cdot n^{O(\ell^{d-1})} = n^{O(\ell^{d-1})} = n^{O(1/\eps^{d-1})}$.
\end{proof}

\subsubsection*{Solving \MIF{} via Geometric Balanced Separator} 

The \emph{ply} of a family of unit balls is the maximum number of balls that cover a single point.
We use the geometric separator theorem of Miller~\etal~\cite{separator} to obtain our exact algorithm parametrized by the output size.

\begin{thm}[Geometric Separator \cite{separator}] \label{thm:mif-sep}
    Let $\lambda\ge 1$ be a constant and let $\cF$ be a set of $k$ unit balls in $\R^d$ of ply at most $\lambda$. Then there exists a $(d{-}1)$-sphere $S$ and a constant $\alpha <1$ depending only on $d$ such that
    \begin{enumerate}[(i)]
        \item at most $\alpha k$ balls of $\cF$ lie strictly inside $S$,
        \item at most $\alpha k$ balls of $\cF$ lie strictly outside $S$, and
        \item at most $O(\lambda^{1/d}k^{1-1/d})$ balls of $\cF$ intersect $S$.
    \end{enumerate}
\end{thm}

\begin{thm}\label{thm:mif-exact-k}
Let $d\ge 2$ be a fixed constant and $\cB$ a set of $n$ unit balls in $\R^d$. Given an integer $k$, one can decide in time $n^{O(k^{1-1/d})}$ whether the intersection graph of $\cB$ has an induced forest of size at least $k$, and if so find one.
\end{thm}

\begin{proof}

We describe a recursive procedure $\textsc{MIF}(\mathcal{B},k)$ that returns an induced forest of size~$k$ in the intersection graph of $\mathcal{B}$ if one exists. If $k \le k_{0}$ for a sufficiently large constant $k_{0}$, we enumerate all $n^{O(1)}$ subsets of $\cB$ of size $k$ and test each for being an induced forest.

If $F$ is a set of unit balls in $\R^d$ whose intersection graph is a forest, then $F$ has ply at most~$2$. By~\Cref{thm:mif-sep}, $F$ admits a balanced $(d-1)$-sphere separator $S$ that is crossed by at most $O(k^{1-1/d})$ balls of $F$. The sphere $S$ can be chosen tangent to at most $d+1$ balls, so there are $n^{O(1)}$ candidate spheres. We enumerate all $n^{O(1)}$ candidate spheres. For each candidate sphere $S$, we enumerate all subsets $X \subseteq \cB$ of size $O(k^{1-1/d})$ consisting of balls intersecting $S$ (discard $X$ if $X$ induces a cycle). There are at most $n^{O(k^{1-1/d})}$ such subsets. Next, we additionally guess the partition of $X$ into tree components of the final forest, to ensure that attaching forests from inside and outside $S$ does not create a cycle through $X$. There are at most $2^{O(k^{1-1/d})}$ such partitions.

Let $\cB_{\mathrm{in}}$ (resp.\ $\cB_{\mathrm{out}}$) be the balls of $\cB\setminus X$ strictly inside (resp.\ outside) $S$. No ball of $\cB_{\mathrm{in}}$ is adjacent to any ball of $B_{\mathrm{out}}$. We remove from each side any ball whose inclusion would close a cycle through $X$ under the guessed partition. For every split $k_1+k_2=k-|X|$ with $k_1\le\alpha k$ and $k_2\le\alpha k$, we recurse
to obtain $F_{\mathrm{in}} = \mathrm{MIF}(B_{\mathrm{in}},k_1)$ and $F_{\mathrm{out}} = \mathrm{MIF}(B_{\mathrm{out}},k_2)$. If both calls succeed we form $F:=F_{\mathrm{in}}\cup X\cup F_{\mathrm{out}}$ and verify whether $F$ is an induced forest of size $k$.

If $\cB$ contains a forest $F^\star$ of size $k$ then $F^\star$ has a separator sphere $S$ with at most $O(k^{1-1/d})$ crossing balls and a balanced split. The algorithm enumerates all candidate $(S,X,\text{partition})$ triples, so some branch discovers $F^\star$. Let $T(n,k)$ be the running time.  At each level the guessing costs $n^{O(k^{1-1/d})}$, after which both subproblems have parameter at most $\alpha k$. Hence $T(n,k) \le n^{O(k^{1-1/d})} \cdot 2\,T(n,\alpha k)$ which solves to $T(n,k)=n^{O(k^{1-1/d})}$.
\end{proof}

Alternatively, one can bound the $\mathcal{P}$-flattened treewidth by at most the separator size times $\log n$, and then solve \MIF{} in a bounded box of side length $\ell$ in time $n^{O(\ell^{d-1})}$ using the $\mathcal{P}$-flattened treewidth machinery of De~Berg~\etal~\cite{deBergBKMZ20}. Moreover, our proofs of \Cref{thm:mif-exact-k} and \Cref{thm:mif-algo} also apply to \MIM{}. By definition, it suffices to check for an induced matching rather than an induced forest. We state this in the following corollaries.

\begin{cor}\label{thm:mim-algo}
For every fixed $d \geq 2$ and $\eps > 0$, there is a $(1-\eps)$-approximation algorithm for \MIM{} on unit ball (or unit cube) intersection graphs in $\R^d$ running in time $n^{O(1/\eps^{d-1})}$.
\end{cor}

\begin{cor}\label{thm:mim-exact-k}
Let $d\ge 2$ be a fixed constant and $\cB$ a set of $n$ unit balls in $\R^d$. Given an integer $k$, one can decide in time $n^{O(k^{1-1/d})}$ whether the intersection graph of $\cB$ has an induced matching of size at least $k$, and if so find one.
\end{cor}

Note that, by the corresponding version of~\cite[Theorem 2.20]{MarxS14} for \tcsp{}, together with our reductions in the proofs of \Cref{thm:mif-ball} and \Cref{thm:mim-ball}, the running times in \Cref{thm:mif-exact-k} and \Cref{thm:mim-exact-k} are tight under ETH. In particular, for any fixed $d \ge 2$, if there exists an algorithm running in time $f(k)n^{o(k^{1-1/d})}$ for finding an induced forest (or induced matching) of size $k$ in the intersection graph of $n$ open unit balls (or unit cubes) in $\mathbb{R}^d$, then ETH fails.

\subsection{PTAS for \MPS{}}

In this subsection, we present a PTAS for \MPS{} on unit balls in $\R^d$ with running time $n^{O(1/\eps^{d-1})}$, using the shifting technique of Hochbaum and Maass~\cite{HochbaumM85}. The key difference from the \textsc{Maximum Induced Forest} setting is that \MPS{} is a minimization problem. Therefore, instead of removing boundary balls from the optimal solution, we account for optimal piercing points near grid boundaries being shared across neighboring cells. Due to the contribution of the boundary piercing points, the cells side length is chosen differently in this application.

\begin{thm}\label{thm:mps-algo}
For every fixed $d \geq 2$ and $\eps > 0$, there is a $(1+\eps)$-approximation algorithm for \MPS{} on unit balls in $\R^d$ running in time $n^{O(1/\eps^{d-1})}$.
\end{thm}

\begin{proof}
Fix $\eps > 0$. Set the shifting parameter $\ell = \lceil 2d \cdot 2^d / \eps \rceil$. For each $d$-tuple of offsets $\ve{i} = (i_1, \ldots, i_d) \in \{0, \ldots, \ell-1\}^d$, partition~$\R^d$ into axis-aligned cubes of side length~$\ell$:
\[
\text{Cell}_{\ve{j}}^{(\ve{i})} = \prod_{r=1}^{d} \big[i_r + j_r \ell, i_r + (j_r+1)\ell\big), \qquad \ve{j} \in \mathbb{Z}^d.
\]
A point $p \in \R^d$ is a \emph{boundary point} for offset~$\ve{i}$ if, for some coordinate~$r$ and some integer~$j_r$, the $r$-th coordinate of~$p$ is within distance~$1$ of the grid hyperplane $x_r = i_r + j_r \ell$. The following algorithm computes a $(1+\eps)$-approximation for \MPS{} on a set~$\cB$ of unit balls.
 
\begin{enumerate}
\item For each offset $\ve{i} \in \{0, \ldots, \ell-1\}^d$:
  \begin{enumerate}
  \item For each cell~$C$, let $\cB_C$ denote the set of balls in~$\cB$ whose centers lie in~$C$.
  \item Solve \MPS{} on $\cB_C$ exactly using \Cref{lem:mps-exact}. Let $S_C$ denote the resulting piercing set.
  \item Let $S_{\ve{i}} = \bigcup_C S_C$.
  \end{enumerate}
\item Return the smallest $S_{\ve{i}}$ over all offsets.
\end{enumerate}

Every ball $B \in \cB$ has its center in a unique cell~$C$, and $S_C$ pierces all balls in $\cB_C$. Hence $S_{\ve{i}}$ is a valid piercing set for all of~$\cB$. We now show that some offset yields a near-optimal solution.
 
\begin{lem}\label{lem:shift-mps}
    There exists an offset $\ve{i}^\star$ such that $|S_{\ve{i}^\star}| \leq (1+\eps) \OPT(\cB)$.
\end{lem}

\begin{proof}
    For a grid cell $C$, define the \emph{expanded cell} $C^{+} = \{x \in \R^d : \mathrm{dist}(x, C) \leq 1\}$, i.e., the Minkowski sum of~$C$ with the unit ball centered at the origin. Let $P^\star$ be an optimal piercing set with $|P^\star| = \OPT(\cB)$. We classify each point of~$P^\star$ according to whether it is a boundary point for the given offset.
    
    Consider a cell $C$ and a unit ball $B \in \cB_C$ centered at $c$, so $c \in C$. Since $P^\star$ pierces $B$, there exists $p \in P^\star$ with $\|p - c\|| \leq 1/2$. In particular, $p \in C^{+}$. Hence the set $P^\star \cap C^{+}$ is a feasible piercing set for $\cB_C$. Since $S_C$ is an optimal piercing set for $\cB_C$, we have $|S_C| \leq |P^\star \cap C^{+}|$, and therefore
    \[
    |S_{\ve{i}}| = \sum_C |S_C| \leq \sum_C |P^\star \cap C^{+}|.
    \]
     
    If $p \in P^\star$ is a non-boundary point (i.e., every coordinate of~$p$ has distance greater than~$1$ from all grid hyperplanes), then $p$ lies in exactly one cell~$C_p$ and belongs to $C_p^{+}$ only. If $p$ is a boundary point, it belongs to at most $2^d$ expanded cells. Denoting the set of boundary points of~$P^\star$ for offset~$\ve{i}$ by $P^\star_{\mathrm{bd}}$, we obtain
    \[
        \sum_C |P^\star \cap C^{+}| \le \bigl|P^\star \setminus P^\star_{\mathrm{bd}}\bigr| + 2^d \cdot \bigl|P^\star_{\mathrm{bd}}\bigr| = \OPT + (2^d - 1)\cdot |P^\star_{\mathrm{bd}}|.
    \]

    Fix a coordinate $r \in \{1, \ldots, d\}$. For each value $i_r \in \{0, \ldots, \ell-1\}$, a point $p \in P^\star$ is boundary in coordinate~$r$ for offset~$i_r$ if and only if the $r$-th coordinate of~$p$ is within distance~$1$ of some grid hyperplane $x_r = i_r + j_r \ell$. Since grid hyperplanes are spaced $\ell \geq 2$ apart, the interval of width $2$ around any fixed coordinate value meets at most $2$ grid lines, corresponding to at most $2$ offsets. Let $P^\star_{r,i_r}$ be the set of boundary points in the $r$-th coordinate for the offset $i_r$. Hence $\sum_{i_r = 0}^{\ell-1} \bigl|P^\star_{ r,i_r}\bigr| \leq 2 \, \OPT(\cB)$ and $\min_{i_r} \bigl|P^\star_{r, i_r}\bigr| \le 2\,\OPT(\cB)/\ell$. Taking a union bound over all $d$~coordinates, the best joint offset $\ve{i}^\star$ satisfies
    \[
    |P^\star_{\mathrm{bd}}| \le d \cdot \frac{2\,\OPT(\cB)}{\ell} \le \frac{\eps}{2^d - 1} \OPT(\cB),
    \]
    where the last inequality uses $\ell \geq 2d(2^d - 1)/\eps$. Therefore,
    \[
    |S_{\ve{i}^\star}| \le \OPT(\cB) + (2^d - 1)\cdot \frac{\eps}{2^d - 1} \OPT(\cB) = (1+\eps)\OPT(\cB).
    \]
\end{proof}
Using \Cref{lem:mps-exact} (or \Cref{thm:mps-exact-k}), the running time analysis is identical to the one in the proof of \Cref{thm:mif-algo}.
\end{proof}

\subsubsection*{Solving \MPS{} via Geometric Separator}

The following result is a corollary of an algorithm of Agarwal and Procopiuc~\cite{AgarwalP02} for the $k$-center problem in $\R^d$
under $L_\infty$ and $L_2$-metrics.

\begin{thm}[Corollary of~\cite{AgarwalP02}]\label{thm:mps-exact-k}
Let $d\ge 2$ be a fixed constant and $\cB$ a set of $n$ unit balls (unit cubes) in $\R^d$. Given an integer $k$, one can decide in time $n^{O(k^{1-1/d})}$ whether $\cB$ has a piercing set of size at most~$k$, and if so find one.
\end{thm}

For completeness, we solve the exact subproblems within a bounded box simply using a separator-based recursive approach. We find a  separator of size $O(\ell^{d-1})$, enumerate all ways the optimal solution can interact with the separator, and recurse on both sides.

\begin{lem}\label{lem:mps-exact}
Let $d \geq 2$ be a fixed constant. Let $\cB$ be a set of $n$ unit balls in $\R^d$ whose centers lie in a cube $[0, \ell]^d$. Then \MPS{} on~$\cB$ can be solved in time $n^{O(\ell^{d-1})}$.
\end{lem}

\begin{proof}
    Consider the arrangement $\cA$ of $n$ balls in~$\cB$. For fixed $d$, this arrangement has $n^{O(1)}$ cells and can be computed in polynomial time.
    We proceed by divide-and-conquer on~$\ell$, always splitting along the longest side of the bounding box. When $\ell = O(1)$ the optimal piercing set has constant size (a grid of $O(\ell^d)=O(1)$ points with spacing $1/2\sqrt{d}$ already pierces every ball), so brute-force over $n^{O(1)}$ cells of $\cA$ suffices.
    
    For the recursive step, let $\ell_j$ be the longest side of the current bounding box with side lengths $\ell_1 \times \ell_2 \times \cdots \times \ell_d$. We define $m := \ell_j/2$. Let $\cS$ be the set of all the ball in $\cB$ with center in the slab~$\{x \in \R^d : |x(j) - m| \leq 1\}$. Again a grid with spacing $1/2\sqrt{d}$ on this slab yields $q = O(\ell^{d-1})$ points that pierce every ball in $\cS$, so any optimal solution uses at most $q$ to pierce the balls in~$\cS$. An optimal piercing set may be assumed to contain at most one representative point per cell of the arrangement~$\cA$. We choose the candidate piercing sets for $\cS$ from the cell representatives of~$\cA$. Each candidate $P_{\cS}$ is a subset of at most $q = O(\ell^{d-1})$ representatives that pierces every ball in~$\cS$. The number of candidate sets is at most $n^{O(\ell^{d-1})}$.
    For each candidate piercing set $P_\cS$ of $\cS$, we partition the un-pierced balls into $L_\cS$ and $R_\cS$. Let $c$ be the center of a ball $B$ that is not pierced by $P_\cS$. If $c(j) \leq m$ then we place $B$ in $L_\cS$, otherwise in $R_\cS$. Note that the balls of $L_\cS$ are disjoint from the balls of $R_\cS$. We recurse on $L_\cS$ with the $j$-th side halved to $\ell_j/2$ and on $R_\cS$ likewise, obtaining piercing sets $P_L$ and $P_R$ respectively. Over all candidates $P_\cS$, we return the set $P_\cS \cup P_L \cup P_R$ with smallest cardinality.

    Let $T(\ell_1,\dots,\ell_d)$ be the running time on an instance with bounding box $\ell_1\times\cdots\times\ell_d$ and $\ell_j = \max_i \ell_i$.
    The enumeration costs $n^{O(\ell^{d-1})}$, and the two subproblems each have the $j$-th side halved, so
    \[
      T(\ell_1,\dots,\ell_j,\dots,\ell_d)
      \le n^{O(\ell_j^{d-1})} \cdot 2\,T(\ell_1,\dots,\ell_j/2,\dots,\ell_d).
    \]
    Since we always split the longest side, after at most $d$ consecutive recurrences every dimension has been halved at least once, so the longest side after $d$ levels is at most $\ell/2$. Solving this recurrence implies the total running time of~$n^{O(\ell^{d-1})}$.
\end{proof}


\bibliography{references}

\end{document}